\documentclass[journal,twoside,web]{ieeecolor}
\usepackage{lcsys}
\usepackage{cite}
\usepackage{amsmath,amssymb,amsfonts}
\usepackage{mathtools}
\usepackage{algorithmic}
\usepackage{graphicx}
\usepackage{textcomp}

\usepackage{subcaption}
\usepackage{float}
\usepackage{esvect}
\usepackage{tikz}

\usepackage[colorlinks=true,allcolors=blue]{hyperref}

\makeatletter
\renewcommand*{\eqref}[1]{%
  \hyperref[{#1}]{\textup{\tagform@{\ref*{#1}}}}%
}
\makeatother

\newtheorem{lem}{Lemma}
\newtheorem{thm}{Theorem}

\newtheorem{defn}{Definition}
\newtheorem{rem}{Remark}
\newtheorem{assump}{Assumption}
\newtheorem{problem}{Problem}

\DeclareMathOperator{\diag}{diag}
\DeclareMathOperator{\sgn}{\mathrm{sgn}}

\usetikzlibrary{shapes,arrows}
\tikzstyle{block} = [draw, rectangle, 
    minimum height=0.75cm, minimum width=1cm]

\usetikzlibrary{positioning, arrows.meta}

\def\BibTeX{{\rm B\kern-.05em{\sc i\kern-.025em b}\kern-.08em
    T\kern-.1667em\lower.7ex\hbox{E}\kern-.125emX}}
\begin{document}
    \title{Robust Global Position and Heading Tracking on $\mathrm{SE}(3)$ via Saturated Hybrid Feedback}

\author{Luís Martins, Carlos Cardeira, \IEEEmembership{Senior Member, IEEE}, and Paulo Oliveira, \IEEEmembership{Senior Member, IEEE}
\thanks{The authors acknowledge Fundação para a Ciência e a Tecnologia (FCT) for its financial support via the projects LAETA Base Funding (DOI: 10.54499/UIDB/50022/2020) and LAETA Programatic Funding (DOI: 10.54499/UIDP/50022/2020). Luís Martins holds a PhD scholarship 2022.14126.BD from FCT.}
\thanks{L. Martins and C. Cardeira are with the Institute of Mechanical Engineering, Instituto Superior Técnico, Universidade de Lisboa, Lisboa, Portugal (e-mails: luis.cunha.martins@tecnico.ulisboa.pt, carlos.cardeira@tecnico.ulisboa.pt).} 
\thanks{P. Oliveira is with the Institute of Mechanical Engineering and the Institute for Systems and Robotics, Instituto Superior Técnico, Universidade de Lisboa, Lisboa, Portugal (e-mail: paulo.j.oliveira@tecnico.ulisboa.pt).}}

\maketitle
\thispagestyle{empty} 

\begin{abstract}
This letter presents a novel control solution to the robust global position and heading tracking problem for underactuated vehicles, equipped with single-axis thrust and full torque actuation, operating under strict, user-defined actuation limits. The architecture features a saturated position tracking controller augmented with two first-order filters. This formulation ensures the boundedness of the first and second derivatives, yielding less conservative bounds and systematically generating bounded attitude references whose limits are easily tuned via design parameters. To track these dynamic references, the inner loop comprises a saturated, modified Rodrigues parameter (MRP)-based controller paired with a hybrid dynamic path-lifting mechanism. This approach allows the attitude tracking law to be designed on a covering space of the configuration manifold. By leveraging a stability equivalence framework, the methodology establishes that the resulting interconnected system achieves robust global asymptotic and semi-global exponential tracking on $\mathrm{SE}(3)$, while complying with user-defined input saturation bounds. Numerical simulations validate the proposed solution.
\end{abstract}


\section{Introduction}
\label{section1}

\subsection{Motivation and Literature Review}

\par \IEEEPARstart{T}{rajectory} tracking for autonomous systems is a fundamental challenge in aerial, marine, and space domains. This problem is exceptionally demanding for underactuated vehicles, where the degrees of freedom outnumber the available control inputs \cite{aguiar2007trajectory}. Compounding this structural deficit, actuator saturation emerges as an inescapable physical nonlinearity that easily degrades performance or induces instability if ignored. Addressing the intersection of underactuation and strict input limits is therefore critical, driving the ongoing demand for robust, globally stable tracking architectures.

\par The topology of the rigid-body motion configuration space, the special Euclidean group $\mathrm{SE}(3)$, fundamentally constrains the trajectory tracking problem. Inheriting the non-contractibility of the rotation group $\mathrm{SO}(3)$, $\mathrm{SE}(3)$ is not homeomorphic to any Euclidean space \cite{BhatBernstein2000}, precluding global stabilization via continuous feedback. Furthermore, discontinuous controllers fail to provide robust global stability, as arbitrarily small disturbances can trap the closed-loop solutions in undesired regions of the state space \cite{mayhew2011topological}. These fundamental obstructions motivate the adoption of hybrid control schemes \cite{casau2015, Naldi2017, wang2021hybrid, basso2022globaluav, martins2023robust, martins2024} over traditional continuous \cite{invernizzi2018trajectory} or discontinuous \cite{aguiar2007trajectory} approaches for tracking on $\mathrm{SE}(3)$.

\par While existing hybrid architectures \cite{casau2015, Naldi2017, wang2021hybrid, basso2022globaluav, martins2023robust} achieve global asymptotic tracking for underactuated vehicles, these methodologies predominantly restrict saturation handling to the thrust input. Standard hierarchical structures encounter a fundamental singularity when the commanded thrust vanishes, losing control authority and precluding a well-defined attitude reference. Driven to avoid this singularity, the literature overwhelmingly prioritizes thrust constraints, frequently overlooking torque saturation and thus leaving a critical gap. To bridge this gap, the approach in \cite{martins2024} introduces a saturated hybrid control solution for quadrotors that enforces strict, \textit{a priori} bounds on both inputs. Relying on nested saturation functions to generate bounded attitude references alongside a saturated MRP-based inner loop, the architecture guarantees global asymptotic tracking.
\vspace{-0.1500cm}
\subsection{Contributions}
\vspace{-0.1500cm}
Motivated by this critical gap, this work proposes a novel control strategy for robust global position and heading tracking of underactuated vehicles driven by single-axis thrust and full torque inputs, subject to strict, user-defined bounds. The solution comprises a sufficiently smooth saturated outer-loop controller that globally asymptotically and semi-globally exponential stabilizes the position tracking dynamics. By replacing the nested saturation function approach of \cite{martins2024} with first-order filters augmentation, the proposed architecture yields significantly simpler expressions for the position control law derivatives. The design relies on a forward-projected coordinate transformation that geometrically decouples the filter dynamics, enabling the controller to be designed under the assumption of ideal, infinite-bandwidth. This augmentation ensures analytically tractable bounds that are easily adjustable via design parameters, directly facilitating the \textit{a priori} verification of trajectory feasibility. Furthermore, adopting the attitude reference from \cite{basso2022globaluav} ensures unrestricted heading tracking, overcoming the inherent geometric coupling in \cite{martins2024} that restricted exact yaw tracking to zero roll and pitch conditions. These commands are tracked by a saturated inner loop developed in the MRP space, whose stability guarantees rigorously translate to the rotation group via a hybrid dynamic path-lifting algorithm and a stability equivalence framework. The strategy in \cite{martins2024} relies on small-gain arguments to establish the stability of the interconnected system, an approach that inherently complicates the attainment of an exponentially stable result. By replacing this methodology with exponentially decaying Lyapunov functions, this work establishes robust global asymptotic and semi-global exponential tracking under strict actuation limits. Furthermore, the resulting control scheme is nominally robust against perturbations, including external disturbances, parametric uncertainties, and measurement noise. To the best of the authors' knowledge, this work introduces the first robust, semi-globally exponentially stable tracking solution for underactuated thrust-propelled vehicles subject to simultaneous thrust and torque saturation.

\subsection{Organization}
\par The remainder of this letter is organized as follows. \hyperref[section2]{Section~\ref*{section2}} establishes the notation and mathematical preliminaries. \hyperref[section3]{Section~\ref*{section3}} defines the trajectory feasibility conditions under actuation bounds and formalizes the control problem. \hyperref[section4]{Section~\ref*{section4}} synthesizes the saturated position controller with bounded derivatives. \hyperref[section5]{Section~\ref*{section5}} formulates the saturated hybrid attitude tracking scheme and proves robust global asymptotic and semi-global exponential tracking on $\mathrm{SE}(3)$ for the resulting interconnected system by leveraging stability equivalence theorems. \hyperref[section6]{Section~\ref*{section6}} validates the architecture through simulations, and \hyperref[section7]{Section~\ref*{section7}} provides concluding remarks.

\vspace{-0.15cm}
\section{Notation and Preliminaries}
\label{section2}

Standard notation is adopted throughout. Let $K\mathbb{B}^n \subset \mathbb{R}^n$ denote the closed origin-centered ball of radius $K>0$, and $K\mathbb{B}_{\infty}^n \triangleq \{\mathbf{x} \in \mathbb{R}^n \colon \|\mathbf{x}\|_\infty \leq K\}$ the corresponding closed ball under the infinity norm. The set of $n \times n$ positive definite matrices is denoted by $\mathbb{R}^{n \times n}_{\succ 0}$. The $n$-dimensional unit sphere is $\mathbb{S}^n = \{\mathbf{x} \in \mathbb{R}^{n+1} \colon \mathbf{x}^{\!\top}\mathbf{x} = 1\}$, and $\bar{\mathbb{R}}^n = \mathbb{R}^n \cup \{\boldsymbol{\infty}\}$ represents the Alexandroff compactification of $\mathbb{R}^n$ \cite[p. 246]{dugundji1966topology}. Furthermore, $\mathbf{I}_n$ is the $n \times n$ identity matrix and $\mathbf{e_i}$ the $i^{\text{th}}$ standard basis vector. The skew-symmetric map $[\cdot]_\times \colon \mathbb{R}^3 \to \{\mathbf{S} \in \mathbb{R}^{3 \times 3} \colon \mathbf{S}^\top = -\mathbf{S}\}$ satisfies $[\boldsymbol{\omega}]_\times\mathbf{s} = \boldsymbol{\omega} \times \mathbf{s}$ \cite{mayhew2013}. For $\mathbf{s} \in \mathbb{R}^n$, $\|\mathbf{s}\|$ is the Euclidean norm, $\|\mathbf{s}\|_{\mathcal{L}_2} \coloneqq \left( \int_{0}^{\infty} \|\mathbf{s}(t)\|^2 \mathrm{d}t\right)^{1/2} < \infty$ the $\mathcal{L}_2$-norm, and $\diag(\mathbf{s}) \triangleq \sum_{i=1}^n (\mathbf{e_i}\mathbf{e_i^{\!\top}})(\mathbf{e_i^{\!\top}}\mathbf{s})$. For a symmetric matrix $\mathbf{A}$, $\lambda_{\min}(\mathbf{A})$ and $\lambda_{\max}(\mathbf{A})$ denote, respectively,  its minimum and maximum eigenvalues. The saturation function here considered is aligned with the following definition:
\begin{defn}
\label{def:SaturationFunction}
The function $\sigma:\mathbb{R} \mapsto \mathbb{R}$ is smooth, odd, strictly increasing, and verifies: (1) $\sigma\left(0\right) = 0$; (2) $s\sigma\left(s\right) > 0\;\; \forall \;\;s \neq 0$; (3) $\lim_{s\to \pm\infty} \sigma\left(s\right) = \pm M$, with $M > 0$; (4) $0 < \dot{\sigma}(s) \leq 1$; (5) $\ddot{\sigma}(s_i) < 0 \;\; \forall \;\; s_i > 0$. \hfill $\square$
\end{defn}
\par Each matrix $\mathbf{R}$ within the attitude configuration manifold $\mathrm{SO}(3) \coloneqq \{\mathbf{R} \in \mathbb{R}^{3 \times 3}: \mathbf{R}^{\!\top}\mathbf{R} = \mathbf{I_3}, \det(\mathbf{R}) = 1\}$ describes a rigid-body rotation from a body-fixed to an inertial frame. The rigid-body pose evolves on the Special Euclidean group $\mathrm{SE}(3) \coloneqq \mathbb{R}^3 \times \mathrm{SO}(3)$. The attitude can equivalently be parameterized by a unit quaternion $\mathbf{q} \coloneqq (q_0, \mathbf{q_1}) \in \mathbb{S}^3$, comprising a scalar part $q_0 \in \mathbb{R}$ and a vector part $\mathbf{q_1} \in \mathbb{R}^3$. The map $\mathcal{R}: \mathbb{S}^3 \to \mathrm{SO}(3)$, defined as $\mathcal{R}(\mathbf{q}) = \mathbf{I_3} + 2 q_0[\mathbf{q_1}]_\times + 2[\mathbf{q_1}]_\times^2$ \cite[Eq. 5]{mayhew2013}, yields the corresponding rotation matrix and satisfies $\mathcal{R}(\mathbf{q}) = \mathcal{R}(-\mathbf{q})$. Its double-valued inverse $\mathcal{Q}: \mathrm{SO}(3) \rightrightarrows \mathbb{S}^3$ is characterized by $\mathcal{Q}(\mathbf{R}) = \{\mathbf{q} \in \mathbb{S}^3: \mathcal{R}(\mathbf{q}) = \mathbf{R}\}$. The MRP representation consist of two vectors, $\boldsymbol{\vartheta}, \boldsymbol{\vartheta}^s \in \bar{\mathbb{R}}^3$, derived via the stereographic projection of the unit quaternion \cite{junkins_2009}:
\vspace{-0.1500cm}
\begin{equation*}
\boldsymbol{\vartheta} = \boldsymbol{\varphi}(\mathbf{q}) = \left\{\!\!\begin{array}{cl}
     \frac{\mathbf{q_1}}{(1 + q_0)} \!\!& \!\! \; \textrm{for} \; \mathbf{q} \in \mathbb{S}^3 \setminus \{\mathbf{s}\}\\
     \boldsymbol{\infty} \!\!& \!\! \; \textrm{for} \;  \mathbf{q} = \mathbf{s} 
\end{array}\right.\!\!\!, \quad \boldsymbol{\vartheta}^s = \boldsymbol{\varphi}(-\mathbf{q}),
\vspace{-0.1500cm}
\end{equation*}
where $\mathbf{s} = (-1,0,0,0)$. The map $\boldsymbol{\Upsilon}: \mathbb{\bar{R}}^3 \to \mathbb{\bar{R}}^3$, given by
\begin{equation*}
    \boldsymbol{\vartheta^s} =  \boldsymbol{\Upsilon}\!\left(\boldsymbol{\vartheta}\right) =  \left\{ \begin{array}{cl}
         -\boldsymbol{\vartheta}\|\boldsymbol{\vartheta}\|^{-2}
         & \!\!, \; \textrm{for} \; \boldsymbol{\vartheta} \in \mathbb{R}^3 \setminus \{\boldsymbol{0}\}  \\
         \boldsymbol{\infty} & \!\!, \; \textrm{for} \; \boldsymbol{\vartheta} \in \{\boldsymbol{0}\}  \\
    \boldsymbol{0} & \!\!, \; \textrm{for} \; \boldsymbol{\vartheta} \in \{\boldsymbol{\infty}\}
    \end{array} 
    \right.,
\vspace{-0.1500cm}
\end{equation*}
enables obtaining the shadow MRP from the original MRP. Both MRP triplets satisfy the kinematic expression \cite{junkins_2009}
\vspace{-0.1500cm}
\begin{equation*}
    \boldsymbol{\dot{\vartheta}} \!=\! \mathbf{T}(\boldsymbol{\vartheta})\boldsymbol{\omega} \!=\!\! \left\{\!\!\!\!\begin{array}{cl}
        \frac{(1-\|\boldsymbol{\vartheta}\|^2)\mathbf{I_3} + 2\left[\boldsymbol{\vartheta}\right]_{\!\times} \!\!+ 2 \boldsymbol{\vartheta} \boldsymbol{\vartheta}^{\!\top}\!}{4}\boldsymbol{\omega} &  \!\!\!\!\!\!, \; \textrm{for} \; \boldsymbol{\vartheta} \!\in\! \mathbb{R}^3 \\
         \boldsymbol{\infty} & \!\!\!\!\!\!, \; \textrm{for} \; \boldsymbol{\vartheta} \!\in\! \{\!\boldsymbol{\infty}\!\} 
    \end{array}\right. \!\!\!.
\end{equation*}
The mapping $\mathcal{R}_{\boldsymbol{\vartheta}}\!\left(\boldsymbol{\vartheta}\right):\mathbb{\bar{R}}^3 \to \textrm{SO(3)}$
\begin{equation}
    \label{eq:rotationmatrixMRP}
\mathcal{R}_{\boldsymbol{\vartheta}}\!\left(\boldsymbol{\vartheta}\right) \coloneqq \!\! \left\{ \!\!\!\begin{array}{cl}
         \mathbf{I_3} + \frac{ 8\left[\boldsymbol{\vartheta}\right]_{\!\times}^2 - 4(1 - \|\boldsymbol{\vartheta}\|^2)\!\left[\boldsymbol{\vartheta}\right]_{\!\times}}{(1 + \|\boldsymbol{\vartheta}\|^2)^{2}} &  \!\!\!\!, \; \textrm{for} \; \boldsymbol{\vartheta} \in \mathbb{R}^3\\
         \mathbf{I_3} & \!\!\!\!, \; \textrm{for} \; \boldsymbol{\vartheta} \in \{\boldsymbol{\infty}\}
    \end{array}\right.
\end{equation}
links a given $\boldsymbol{\vartheta}$ to the equivalent $\mathbf{R}$ and satisfies $\mathcal{R}_{\boldsymbol{\vartheta}}(\boldsymbol{\vartheta}) = \mathcal{R}_{\boldsymbol{\vartheta}}(\boldsymbol{\vartheta^s})$. For more insights on MRP, please see \cite{junkins_2009}.

\section{Problem Statement}
\label{section3}



Consider the underactuated rigid-body dynamics
\begin{subequations}
    \label{eq:FullDynamics}
    \begin{equation}
\label{eq:PositionDynamics}
    \mathbf{\dot{p}} = \mathbf{v}, \quad \mathbf{\dot{v}} = - g\mathbf{e_3} + \mathbf{R}\mathbf{e}_3\tfrac{T}{m},
\end{equation}
\begin{equation}
\label{eq:AttitudeDynamics}
    \mathbf{\dot{R}} = \mathbf{R}\left[\boldsymbol{\omega}\right]_{\!\times}, \quad \mathbf{J}\boldsymbol{\dot{\omega}} = \left[\mathbf{J}\boldsymbol{\omega}\right]_{\!\times} \boldsymbol{\omega} + \boldsymbol{\tau},
\end{equation}
\end{subequations}
where $(\mathbf{p}, \mathbf{v}) \in \mathbb{R}^3 \times \mathbb{R}^3$ are the inertial position and velocity, $\mathbf{R} \in \mathrm{SO}(3)$ is the body-to-inertial rotation matrix, and $\boldsymbol{\omega} \in \mathbb{R}^3$ is the body-frame angular velocity. The constants $m, g \in \mathbb{R}_{>0}$ and $\mathbf{J} \in \mathbb{R}^{3\times3}$ represent the mass, gravitational acceleration, and positive-definite diagonal inertia matrix. Physical actuation limits constrain the inputs, such that the thrust magnitude $T$ and the control torque $\boldsymbol{\tau} \coloneqq (\tau_\varphi, \tau_\theta, \tau_\psi)$ reside in the bounded sets $\boldsymbol{\Omega}_{T} \subset \mathbb{R}_{>0}$ and $\boldsymbol{\Omega_\tau} \subset \mathbb{R}^3$, defined as
\begin{equation*}
    \begin{aligned}
\boldsymbol{\Omega}_{T} & \coloneqq \{T \in \mathbb{R}_{> 0} \colon T \leq T_{\max}\}, \\
\boldsymbol{\Omega_\tau} &\coloneqq \{\boldsymbol{\tau} \in \mathbb{R}^3 \colon |\mathbf{e_i^{\!\top}}\boldsymbol{\tau}| \leq |\mathbf{e_i^{\!\top}}\boldsymbol{\tau_{\max}}| \quad \forall \quad i \in \{1,2,3\}\},
\end{aligned}
\end{equation*}
where $T_{\max} \in \mathbb{R}_{>0}$ denotes the maximum available thrust, and $\boldsymbol{\tau_{\max}} \coloneqq (\tau_{\varphi_{\max}}, \tau_{\theta_{\max}}, \tau_{\psi_{\max}}) \in \mathbb{R}^3_{>0}$ represents the respective component-wise maximum torque bounds.
\par Let the reference trajectory map $\mathbf{r} \colon \mathbb{R}_{\geq 0} \to \mathbb{R}^{12} \times \mathbb{S}^2 \times \mathbb{R}^2$ be defined as $\mathbf{r}(t) \coloneqq (\mathbf{p_d}, \mathbf{\dot{p}_d}, \mathbf{\ddot{p}_d}, \mathbf{p_d}^{(3)}, \boldsymbol{\nu_d}, \boldsymbol{\dot{\nu}_d})(t)$, encompassing the desired position $\mathbf{p_d} \in \mathbb{R}^3$, heading relative to the inertial frame $\boldsymbol{\nu_d} \in \mathbb{S}^2$, and their respective derivatives. To guarantee feasibility under actuation limits, the trajectory must satisfy the following assumption.

\begin{assump}
    \label{assump:ReferenceTrajectory}
    The reference $\mathbf{r}(t)$ lies in the compact set $\boldsymbol{\Omega}_{\mathbf{r}} \subset \mathbb{R}^{3} \times \mathbb{R}^3 \times \boldsymbol{\Omega}_{\mathbf{a}} \times K_{j}\mathbb{B}^3 \times \mathbb{S}^2 \times K_{\dot{\nu}}\mathbb{B}^2,$ with
\begin{equation*}
\boldsymbol{\Omega}_{\mathbf{a}}  \!\coloneqq\! \left\{\mathbf{\ddot{p}_d} \in \mathbb{R}^3 \colon (\mathbf{e_1^{\!\top}}\mathbf{\ddot{p}_d}, \mathbf{e_2^{\!\top}}\mathbf{\ddot{p}_d}) \in K_{a_{1,2}}\mathbb{B}^2, \mathbf{e_3^{\!\top}}\mathbf{\ddot{p}_d} \in K_{a_3}\mathbb{B}\right\},
\end{equation*}
and its evolution is governed by the differential inclusion
\begin{equation}
\label{eq:SystemTrajectory}
\mathbf{\dot{r}} \in \mathbf{F_r}(\mathbf{r}) \coloneqq \big(\mathbf{\dot{p}_d}, \mathbf{\ddot{p}_d}, \mathbf{p_d^{(3)}}, K_{s}\mathbb{B}^3, \boldsymbol{\dot{\nu}}_{\mathbf{d}}, K_{\ddot{\nu}}\mathbb{B}^2\big).
\end{equation}
Furthermore, the following conditions hold
\begin{subequations}
\label{eq:ConditionsTrajectory}
\begin{equation}
			\label{eq:ThrustConditions}
			g > K_{a_3}, \quad T_{\max} > m \left(g + \sqrt{K_{a_{1,2}}^2 + K_{a_3}^2}\right),
		\end{equation}
\begin{equation}
\label{eq:TorqueMaxCondition}
    \mathbf{e_i^{\!\top}}\boldsymbol{\tau_{\max}} > J_{ii} U^*_{\dot{\omega}_d} + \frac{1}{2} \left| J_{jj} - J_{kk} \right| {U^*_{\omega_d}}^2,
\end{equation}
\end{subequations}
for all $i \in \{1,2,3\}$, with $j$ and $k$ being the remaining mutually distinct indices such that $\{j, k\} = \{1, 2, 3\} \setminus \{i\}$, $J_{ii} = \mathbf{e_i^{\!\top}} \mathbf{J} \mathbf{e_i}$ representing the $i$-th diagonal element of the inertia matrix, and 
\begin{equation*}
U^*_{\omega_d} \coloneqq \sqrt{ U_{\dot{\rho}}^2 + U_{\dot{\upsilon}}^2 }, \quad U^*_{\dot{\omega}_d} \coloneqq U_{\dot{\rho}} U_{\dot{\upsilon}} + \sqrt{ U_{\ddot{\rho}}^2 + U_{\ddot{\upsilon}}^2 },
\end{equation*}
where
\begin{equation*}
\begin{aligned}
U_{\dot{\rho}} &\coloneqq \frac{K_j}{g - K_{a_3}}, \quad && U_{\ddot{\rho}} \coloneqq \frac{K_s}{g - K_{a_3}} + 3U_{\dot{\rho}}^2, \\
U_{\dot{\upsilon}} &\coloneqq \frac{U_{\dot{\rho}} + K_{\dot{\nu}}}{L^*}, \quad && U_{\ddot{\upsilon}} \coloneqq \frac{U_{\ddot{\rho}} + K_{\ddot{\nu}} + 2 U_{\dot{\rho}} K_{\dot{\nu}}}{L^*} + 3 U_{\dot{\upsilon}}^2, 
\end{aligned}
\end{equation*}
\vspace{-0.25cm}
\begin{equation*}
    L^* \coloneqq \frac{g - K_{a_3}}{\sqrt{K_{a_{1,2}}^2 + (g - K_{a_3})^2}}.
\vspace{-0.5cm}
\end{equation*}
$\hfill \square$
\end{assump}
\par Maximal solutions $\mathbf{r}(t)$ to \eqref{eq:SystemTrajectory} are guaranteed to exist and are complete. Furthermore, along any such trajectory, the signals $\mathbf{p_d}^{(3)}$ and $\boldsymbol{\nu}_{\mathbf{d}}$ are Lipschitz continuous with respective constants $K_s$ and $K_{\dot{\nu}}$.

\par To map the translational objectives of the underactuated system into attitude commands, the desired attitude $\mathbf{R_d} \in \mathrm{SO}(3)$ is algebraically extracted from the thrust direction $\boldsymbol{\rho} \in \mathbb{S}^2$ and the desired heading \cite{basso2022globaluav}. The rotation matrix is defined by the orthonormal basis
\begin{equation}
\label{eq:DesiredRotationMatrix}
    \mathbf{R_d}(\boldsymbol{\rho}, \boldsymbol{\upsilon}) = \begin{bmatrix}
        \boldsymbol{\upsilon} & [\boldsymbol{\rho}]_{\times}\boldsymbol{\upsilon} & \boldsymbol{\rho}
    \end{bmatrix},
\end{equation}
where $\boldsymbol{\upsilon} \in \mathbb{S}^2$ is explicitly computed as $\boldsymbol{\upsilon} \coloneqq \frac{\boldsymbol{\varpi}}{\|\boldsymbol{\varpi}\|}$, with
\begin{equation*}
    \boldsymbol{\varpi} \coloneqq \sgn(\mathbf{e_3^{\!\top}} \boldsymbol{\rho}) \begin{bmatrix}
        (\mathbf{e_3^{\!\top}} \boldsymbol{\rho})\boldsymbol{\nu}_{\mathbf{d}} \\
        -\boldsymbol{\rho}^\top\boldsymbol{\nu^*}
    \end{bmatrix}, \quad \text{and} \quad \boldsymbol{\nu^*} \coloneqq \begin{bmatrix}
        \boldsymbol{\nu}_{\mathbf{d}} \\ 0
    \end{bmatrix}.
\end{equation*}
This formulation ensures the strict well-posedness of $\mathbf{R_d}$ for all $\boldsymbol{\rho}, \boldsymbol{\upsilon} \in \mathbb{S}^2$, provided $\mathbf{e_3^{\!\top}}\boldsymbol{\rho} \neq 0$. The vector $\boldsymbol{\upsilon}$ corresponds to the desired body-fixed $x$-axis described in the inertial frame, obtained by projecting the augmented desired heading $\boldsymbol{\nu^*}$ onto the plane orthogonal to $\boldsymbol{\rho}$. The strictly positive proportionality between the first two components of $\boldsymbol{\upsilon}$ and $\boldsymbol{\nu_d}$ ensures that the projection of $\boldsymbol{\upsilon}$ onto the inertial horizontal plane aligns perfectly with the desired heading. This directly contrasts with the approaches in \cite{casau2015, martins2024}, where tracking guarantees apply exclusively to the projection of the desired body-fixed $x$-axis onto the plane orthogonal to $\boldsymbol{\rho}$. \hyperref[prob:Problem]{Problem~\ref*{prob:Problem}} formulates the control objective driving this work. 

\begin{problem}
\label{prob:Problem}
Let $\mathbf{x}\coloneqq (\mathbf{p}, \mathbf{v}, \mathbf{R}, \boldsymbol{\omega}) \in \boldsymbol{\chi} \coloneqq \mathbb{R}^3 \times \mathbb{R}^3 \times \mathrm{SO}(3) \times \mathbb{R}^{3}$. Design bounded inputs $T \in \boldsymbol{\Omega}_T$ and $\boldsymbol{\tau} \in \boldsymbol{\Omega_\tau}$ such that the compact set 
\vspace{-0.150cm}
\begin{equation*}
    \mathcal{A} = \{(\mathbf{r}, \mathbf{x}) \in \boldsymbol{\Omega}_{\mathbf{r}} \times \boldsymbol{\chi}: \mathbf{p} = \mathbf{p_d}, \mathbf{v} = \mathbf{v_d}, \mathbf{R} = \mathbf{R_d}, \boldsymbol{\omega} = \boldsymbol{\omega}_{\mathbf{d}}\} \vspace{-0.1500cm}
\end{equation*}
is robustly globally asymptotically and semi-globally exponentially stable for the system \eqref{eq:FullDynamics}.
\hfill $\square$ 
\end{problem}

\section{Filtered Saturated Position Tracking}
\label{section4}

\par The control design initiates with the development of a filtered saturated feedback law for the position tracking system. By relying on a smooth saturation function and integrating two low-pass filters, the approach enforces boundedness on the control signal and its first two derivatives, providing the necessary foundation for generating feasible angular velocity and acceleration references. This architecture allows the filter gains to serve as direct tuning parameters, offering an intuitive mechanism for adjusting the upper bounds of these references.

\par Let the position and velocity tracking errors $\mathbf{\Tilde{p}}, \mathbf{\Tilde{v}} \in \mathbb{R}^3$ be given by, respectively, $\mathbf{\Tilde{p}} \coloneqq \mathbf{p} - \mathbf{p_d}$ and $\mathbf{\Tilde{v}} \coloneqq \mathbf{v} - \mathbf{\dot{p}_d}$. The kinematics and dynamics of the position tracking errors are given by
\begin{equation*}
    \dot{\tilde{\mathbf{p}}} = \tilde{\mathbf{v}}, \quad \dot{\tilde{\mathbf{v}}} = \mathbf{u} - g\mathbf{e}_3 - \ddot{\mathbf{p}}_d,
\end{equation*}
where $\mathbf{u} \coloneqq \mathbf{R}\frac{T}{m}\mathbf{e}_3 \in \mathbb{R}^3$ is the designated control vector. From this formulation, the thrust magnitude $T \coloneqq m\|\mathbf{u}\|$ and direction $\mathbf{R}\mathbf{e}_3 \coloneqq \frac{\mathbf{u}}{\|\mathbf{u}\|}$ are extracted, provided $\|\mathbf{u}\| > 0$. The control law is then partitioned into a filtered feedback term $\mathbf{u_f} \in \mathbb{R}^3$ and feedforward compensation:
\begin{equation}
\label{eq:PositionControlLaw}
    \mathbf{u} \coloneqq \mathbf{u_f} + g\mathbf{e}_3 + \ddot{\mathbf{p}}_d.
\end{equation}
The term $\mathbf{u_f} \in \mathbb{R}^3$ is the output of dual first-order filters applied to the difference between the primary controller $\mathbf{u_p} \in \mathbb{R}^3$ and the feedforward terms, governed by the dynamics
\begin{equation*}
    \mathbf{\dot{u}_f} = -k_f(\mathbf{u_f} - \mathbf{u_s}), \quad \dot{\mathbf{u}}_s = -k_s(\mathbf{u_s} - (\mathbf{u_p} - g\mathbf{e_3} - \mathbf{\ddot{p_d}})).
\end{equation*}
Here, $\mathbf{u_s}$ acts as the internal filter state and $k_f, k_s \in \mathbb{R}_{>0}$ are the filter gains, dictating the bandwidths. Let the primary controller satisfy
\begin{equation*}
    \mathbf{u_p} \coloneqq \mathbf{\bar{u}_p} + g\mathbf{e_3} + \mathbf{\ddot{p}_d},
\end{equation*}
with $\mathbf{\bar{u}_p} \in \mathbb{R}^3$ representing a saturated state-dependent feedback law to be designed. By defining $\mathbf{x_f} \coloneqq (\mathbf{u_f}, \mathbf{u_s}) \in \mathbb{R}^6$, the latter dynamics can be compactly written as follows
\begin{equation}
\label{eq:FilterDynamics}
    \mathbf{\dot{x}_f} = \mathbf{F_f}(\mathbf{x_f}, \mathbf{\bar{u}_p}) = \left(\begin{array}{cc}
         -k_f(\mathbf{u_f} - \mathbf{u_s}) \\
         -k_s(\mathbf{u_s} - \mathbf{\bar{u}_p})
    \end{array} \right).
\end{equation}
The primary objective is to design the bounded input $\mathbf{\bar{u}_p}$ such that the resulting controller $\mathbf{u}$ and its first derivative are Lipschitz continuous. Let $\mathbf{\Tilde{x}_p} \coloneqq (\mathbf{\Tilde{p}}, \mathbf{\Tilde{v}}, \mathbf{x_f}) \in \mathbb{R}^{12}$ denote the augmented position tracking state vector. The feedback law $\mathbf{u}$ leads to the following closed-loop tracking dynamics
\begin{equation}
\label{eq:PositionTrackingDynamics}
    \mathbf{\dot{\Tilde{x}}_p} \coloneqq \left( \begin{array}{c}
        \mathbf{F_p}(\mathbf{\Tilde{v}}, \mathbf{u_f})  \\ 
        \mathbf{F_f}(\mathbf{x_f}, \mathbf{\bar{u}_p})
    \end{array} \right) =  \left(\begin{array}{c}
         \mathbf{\Tilde{v}} \\
         \mathbf{u_f} \\
         \mathbf{F_f}(\mathbf{x_f}, \mathbf{\bar{u}_p})
    \end{array}\right) 
\end{equation}
It was demonstrated in \cite{martins2024} that a smooth, saturated linear feedback renders the nominal double integrator globally asymptotically stable and semi-globally exponentially stable. However, the inclusion of the filter dynamics induces phase lag, which generally invalidates these stability guarantees. To systematically circumvent this limitation and leverage the nominal results in this setting, a change of coordinates is proposed to structurally decouple the tracking dynamics from the filter states. Consider the linear mapping
\begin{equation}
\label{eq:Zcoordinates}
    \mathbf{z} = \boldsymbol{\Phi}_{\mathbf{p}} \mathbf{\Tilde{x}_p}, 
\end{equation}
where $\mathbf{z} \coloneqq (\mathbf{z_1}, \mathbf{z_2}, \mathbf{u_f}, \mathbf{u_s}) \in \mathbb{R}^{12}$ and $\boldsymbol{\Phi}_{\mathbf{p}} \in \mathbb{R}^{12 \times 12}$ is defined by 
\begin{equation*}
    \boldsymbol{\Phi}_{\mathbf{p}} = \left[ \begin{array}{cccc}
        \mathbf{I_3} & \left( \frac{1}{k_f} + \frac{1}{k_s} \right) \mathbf{I_3} & \frac{1}{k_s k_f}\mathbf{I_3} & \boldsymbol{0} \\
         \boldsymbol{0} & \mathbf{I_3} & \frac{1}{k_f} \mathbf{I_3} & \frac{1}{k_s} \mathbf{I_3} \\ 
         \boldsymbol{0} & \boldsymbol{0} & \mathbf{I_3} & \boldsymbol{0} \\
         \boldsymbol{0} &\boldsymbol{0} &\boldsymbol{0} & \mathbf{I_3}
    \end{array} \right]. 
\end{equation*}
\par The coordinate transformation provides a critical structural simplification in the proposed control synthesis. Characterized as a constant, non-singular, upper-triangular matrix with a unit determinant, $\boldsymbol{\Phi}_{\mathbf{p}}$ establishes a strict linear isomorphism between the physical and transformed state spaces. This global diffeomorphism analytically forward-propagates the cascaded filter states into the kinematic coordinates, algebraically transforming the coupled dynamics into a Brunovsky canonical form $\mathbf{z_1}$ and $\mathbf{z_2}$ dynamics driven by $\mathbf{\bar{u}_p}$:
\begin{equation*}
    \left[\begin{array}{c}
         \mathbf{\dot{z}_1}   \\
         \mathbf{\dot{z}_2} 
    \end{array}\right] = \mathbf{F_z}(\mathbf{z_1}, \mathbf{z_2}, \mathbf{\bar{u}_p}) = \left[ \begin{array}{c}
         \mathbf{z_2}  \\
         \mathbf{\bar{u}_p}
    \end{array}\right].
\end{equation*}
By absorbing the inherent phase lag into a forward-projected coordinate, the transformation effectively isolates the filter dynamics from the resulting double integrator, as depicted in \hyperref[fig:PhysicalTransformedSystem]{Fig.~\ref{fig:PhysicalTransformedSystem}}. This geometric decoupling permits the controller to be designed under the assumption of ideal, infinite-bandwidth actuation, as the finite filter constraints are structurally encapsulated within $\boldsymbol{\Phi}_{\mathbf{p}}$. Specifically, the coordinate $\mathbf{z_1}$ represents a dynamically projected position that compensates for cascaded delays, inherently quantifying the residual phase lag accumulated within the finite filters, whereas $\mathbf{z_2}$ denotes the effective velocity driven by the input $\mathbf{u_p}$. Consequently, the dynamic independence of the states streamlines the design, with the global diffeomorphism formally guaranteeing that global asymptotic stability in the $\mathbf{z}$-domain maps identically to the physical domain. In this way, the results proposed in \cite{martins2024} can be applied. Accordingly, let the feedback law $\mathbf{\bar{u}_p}$ be defined as 
\begin{equation*}
    \mathbf{\bar{u}_p} \coloneqq -\boldsymbol{\sigma}_{\mathbf{p}}(k_p \mathbf{z_1} + k_v \mathbf{z_2}),
\end{equation*}
where $\boldsymbol{\sigma}_{\mathbf{p}}:\mathbb{R}^3 \to \mathbb{R}^3$ denotes a saturation function whose saturation level $M_p \in \mathbb{R}_{>0}$ verifies $M_p < g - K_{a_3}$, and $k_p, k_v \in \mathbb{R}_{>0}$ are the feedback gains. \hyperref[thm:StabilityPositionSystem]{Theorem~\ref{thm:StabilityPositionSystem}} formally asserts the tracking properties achieved under the proposed controller.

\begin{thm}
\label{thm:StabilityPositionSystem}

Suppose the conditions of \hyperref[assump:ReferenceTrajectory]{Assumption~\ref*{assump:ReferenceTrajectory}} hold for all $t \geq 0$. Then, the origin $\mathbf{\Tilde{x}_p} = \boldsymbol{0}$ of the tracking dynamics \eqref{eq:PositionTrackingDynamics} is globally asymptotically stable and semi-globally exponentially stable. Furthermore, along any solution $\mathbf{\Tilde{x}_p}: \mathbb{R}_{\geq 0} \to \mathbb{R}^{12}$ to \eqref{eq:PositionTrackingDynamics} with the filter states subject to the initial conditions $\mathbf{x_f}(0) \in \mathbb{B}^6_{\infty}$, the resulting thrust input strictly satisfies $T \in [T_{\min}, T_{\max}]$ with
\begin{subequations}
    \label{eq:BoundsT}
\begin{equation}
    \label{eq:LowerBoundT}
    T_{\min} = m \left(g - M_p - K_{a_3}\right) > 0, 
\end{equation}
\begin{equation}
    \label{eq:UpperBoundT}
    T_{\max} = m \Big(\sqrt{3}M_p + \sqrt{K_{a_{1,2}}^2 + (K_{a_3} + g)^2} \Big).
\end{equation}
\end{subequations}
\end{thm}
\begin{proof}
    Consider the transformed dynamics 
    \begin{equation}
    \label{eq:TransformedPositionDynamics}
        \mathbf{\dot{z}} = \boldsymbol{\Phi}_{\mathbf{p}} \mathbf{\dot{\Tilde{x}}_p} = \left[ \begin{array}{c}
             \mathbf{F_z}(\mathbf{z_1}, \mathbf{z_2}, \mathbf{u_p})  \\
             \mathbf{F_f}(\mathbf{x_f}, \mathbf{u_p})
        \end{array} \right]
    \end{equation}
and the Lyapunov function candidate $V_1 \colon\mathbb{R}^{12} \to \mathbb{R}_{\geq 0}$ given by
\begin{equation}
\label{eq:LyapunovFunctionV}
V_1(\mathbf{z}) \coloneqq  \frac{1}{k_v} V_p(\mathbf{z_1}, \mathbf{z_2}) + \frac{1}{2}V_f(\mathbf{x_f}),
\end{equation}
where, similar to \cite[Theorem 1]{martins2024}, $V_p \colon \mathbb{R}^6 \to \mathbb{R}_{\geq 0}$ has the form
\begin{equation*}
    V_p(\mathbf{z_1}, \mathbf{z_2}) \coloneqq k_p \|\mathbf{z_2}\|^2 +
\int_{0}^{k_p \mathbf{z_1} + k_v \mathbf{z_2}} \!\!\!\!\!\!\!\!\!\!\!\!\!\!\!\!\boldsymbol{\sigma}_{\mathbf{p}}(\boldsymbol{\mu})^{\top} \!\; \mathrm{d}\boldsymbol{\mu} \!+\!\!  \int_{0}^{k_p \mathbf{z_1}} \!\!\!\!\!\!\!\!\! \boldsymbol{\sigma}_{\mathbf{p}}(\boldsymbol{\mu})^\top \!\!\; \mathrm{d}\boldsymbol{\mu} ,
\end{equation*}
and $V_f \colon \mathbb{R}^6 \to \mathbb{R}_{\geq 0}$ is defined as
\begin{equation*}
    V_f(\mathbf{x_f}) \coloneqq \frac{1}{k_f} \|\mathbf{u_f}\|^2 + \frac{1}{k_s}\|\mathbf{u_s}\|^2.
\end{equation*}
\noindent This function is radially unbounded, positive-definite, continuously differentiable, and satisfies the upper bound
\begin{equation*}
    V_1(\mathbf{z}) \leq \! \frac{\|k_p \mathbf{z_1} + k_v \mathbf{z_2}\|^2 \!+\! \|k_p \mathbf{z_1}\|^2 \!+\! 2k_p\|\mathbf{z_2}\|^2}{2 k_v}
     \!+\! \frac{1}{2}V_f(\mathbf{x_f}).
\end{equation*}
Thus, $ V(\mathbf{z}) \!\leq\! \lambda_{\max}\!\left(\mathbf{A_p}\right) \!\|\mathbf{z}\|^2 $, with $\mathbf{A_p} \!\in\! \mathbb{R}^{12\!\times\!12}_{\succ 0}$ given by
\begin{equation*}
     \mathbf{A_p} = \begin{bmatrix}
\frac{k_p^2}{k_v} \mathbf{I_3} & \frac{k_p}{2} \mathbf{I_3} & \boldsymbol{0} & \boldsymbol{0} \\
\frac{k_p}{2} \mathbf{I_3} & \frac{k_v^2 + 2k_p}{2 k_v} \mathbf{I_3} & \boldsymbol{0} & \boldsymbol{0} \\
\boldsymbol{0} & \boldsymbol{0} & \frac{1}{2 k_f} \mathbf{I_3} & \boldsymbol{0} \\
\boldsymbol{0} & \boldsymbol{0} & \boldsymbol{0} & \frac{1}{2 k_s} \mathbf{I_3}
\end{bmatrix}.
\end{equation*}
Along the trajectories of \eqref{eq:TransformedPositionDynamics}, the derivative $\dot{V}_1(\mathbf{z})$ verifies
\begin{equation*}
    \begin{aligned}
        \dot{V}_1(\mathbf{z}) = &  -\frac{k_p}{k_v}(-\mathbf{\bar{u}_p} - \boldsymbol{\sigma}_{\mathbf{p}}(k_p\mathbf{z_1}))^{\!\top}\mathbf{z_2}  -\tfrac{1}{2}\|\mathbf{\bar{u}_p}\|^2\\ & - \begin{bmatrix}
             \mathbf{u_f} \!\!&\!\! \mathbf{u_s} \!\!&\!\! \mathbf{\bar{u}_p}
        \end{bmatrix} \begin{bmatrix}
            \mathbf{I_3} & -\frac{1}{2}\mathbf{I_3} & \boldsymbol{0} \\
            -\frac{1}{2}\mathbf{I_3} & \mathbf{I_3} & -\frac{1}{2}\mathbf{I_3} \\
            \boldsymbol{0} & -\frac{1}{2}\mathbf{I_3} & \frac{1}{2}\mathbf{I_3}
        \end{bmatrix} \begin{bmatrix}
            \mathbf{u_f} \\ \mathbf{u_s} \\ \mathbf{\bar{u}_p}
        \end{bmatrix}.
    \end{aligned}
\end{equation*}
Since the saturation function is strictly increasing, 
\begin{equation*}
    -\tfrac{k_p}{k_v}(-\mathbf{\bar{u}_p} - \boldsymbol{\sigma}_{\mathbf{p}}(k_p\mathbf{z_1}))^{\!\top}\mathbf{z_2} < 0 \quad \textrm{for} \quad \mathbf{z_2} \neq \boldsymbol{0}.
\end{equation*}
Furthermore, based on Sylvester's criterion, the third term is negative definite with respect to $\{\mathbf{z} \in \mathbb{R}^{12} : \mathbf{u_f} = \mathbf{u_s} = \mathbf{\bar{u}_p} = \boldsymbol{0}\}$. 
Consequently, the time derivative satisfies $\dot{V}_p(\mathbf{z}) = -W_p(\mathbf{z})$, where the mapping $W_p : \mathbb{R}^{12} \to \mathbb{R}_{\geq 0}$ is continuous and strictly positive-definite. Thus, it follows from \cite[Theorem 4.9]{khalil_2002} that the equilibrium point $\mathbf{z} = \boldsymbol{0}$ is globally uniformly asymptotically stable for \eqref{eq:TransformedPositionDynamics}. Following an approach similar to \cite[Theorem 1]{martins2024}, it stems from this result that $V_1(\mathbf{z}(t)) \leq V_1(\mathbf{z}(0)) \quad \forall \quad t \in \mathbb{R}_{\geq 0}$. Based on Weierstrass’s extreme value theorem, the continuous functions $\|k_p\mathbf{z_1} + k_v \mathbf{z_2}\|$ and $\|k_p\mathbf{z_1}\|$ have a maximum on the compact set $\mathcal{V}_{\mathbf{1}} \coloneqq \{\mathbf{z} \in \mathbb{R}^{12}: V_1(\mathbf{z}(t)) \leq V_1(\mathbf{z}(0)) \}$. Hence, these signals are strictly bounded by a constant dependent on the initial conditions. Specifically, $\max\{\|k_p\mathbf{z_1} + k_v \mathbf{z_2}\|, \|k_p\mathbf{z_1}\|\} \leq \alpha_p$, where $\alpha_p \in \mathbb{R}_{\geq 0}$ is established over the compact set $\mathcal{V}_{\mathbf{1}}$ via
\begin{equation*}
    \alpha_p \coloneqq \max_{\mathbf{z} \in \mathcal{V}_{\mathbf{1}}} \, \max \big\{ \|k_p\mathbf{z_1} + k_v \mathbf{z_2}\|, \|k_p\mathbf{z_1}\| \big\}.
\end{equation*}
The positive definiteness of $V(\mathbf{z})$ ensures that $\alpha_p = 0$ if and only if $\mathbf{z}(0) = \boldsymbol{0}$, which corresponds to the trivial equilibrium solution $\mathbf{z}(t) = \boldsymbol{0}$. Therefore, to establish the exponential convergence of non-trivial trajectories, the subsequent analysis strictly assumes $\alpha_p \in \mathbb{R}_{>0}$. Leveraging the last property established in \hyperref[def:SaturationFunction]{Definition~\ref*{def:SaturationFunction}}, the following bounds are obtained:
\begin{subequations}
\label{eq:SaturationBoundsPositionStability}
\begin{equation}
    \label{eq:P1fromDef}
    \|\boldsymbol{\sigma}_{\mathbf{p}}(\mathbf{y})\| \geq \frac{\sigma_p(\alpha_p)}{\alpha_p} \|\mathbf{y}\|, \; \forall \; \mathbf{y} \in \{ k_p\mathbf{z}_1 + k_v \mathbf{z}_2, k_p\mathbf{z}_1 \},
\end{equation}
\begin{equation}
    \label{eq:P2fromDef}
    \|\mathbf{\bar{u}_p} + \boldsymbol{\sigma}_{\mathbf{p}}(k_p \mathbf{z_1})\| \geq \sigma'_p(\alpha_p)\|k_v\mathbf{z_2}\|.
\end{equation}
\end{subequations}
Substituting these bounds into the Lyapunov derivative yields
\begin{equation}
\label{eq:V1dotbound}
    \dot{V}_1(\mathbf{z}) \leq -\lambda_{\min}(\mathbf{B_p}(\alpha_p))\|\mathbf{z}\|^2,
\end{equation}
where the symmetric block matrix $\mathbf{B_p}: \mathbb{R}_{>0} \to \mathbb{R}^{4\times4}_{\succ 0}$ is partitioned as
\begin{equation*}
    \mathbf{B_p}(\alpha) = \begin{bmatrix}
       \mathbf{B_1}(\alpha) & \mathbf{B_2}(\alpha) \\
       \mathbf{B_2}^\top(\alpha) & \mathbf{B_3}
    \end{bmatrix},
\end{equation*}
with its respective submatrices defined by
\begin{align*}
    \mathbf{B_1}(\alpha) &= \begin{bmatrix}
         k_p^2 \frac{\sigma_p(\alpha)^2}{2\alpha^{2}} & k_p k_v\frac{\sigma_p(\alpha)^2}{2\alpha^{2}} \\
         k_p k_v\frac{\sigma_p(\alpha)^2}{2\alpha^{2}} & k_v^2\frac{\sigma_p(\alpha)^2}{2\alpha^{2}} + k_p \sigma'_p(\alpha) 
    \end{bmatrix}, \\
    \mathbf{B_2}(\alpha) &= \begin{bmatrix}
        0 & k_p \frac{\sigma_p(\alpha)}{2\alpha} \\
        0 & k_v \frac{\sigma_p(\alpha)}{2\alpha}
    \end{bmatrix}, \quad 
    \mathbf{B_3} = \begin{bmatrix}
        1 & \frac{1}{2} \\
        \frac{1}{2} & 1
    \end{bmatrix}.
\end{align*}
Since $\sigma_p(\cdot)$ is a strictly increasing function, the matrix $\mathbf{C}(\alpha)$ is positive definite for any given $\alpha \in \mathbb{R}_{>0}$. Thus, the derivative satisfies
\begin{equation}
    \label{eq:VdotV}
    \dot{V}_1(\mathbf{z}) \leq -\frac{\lambda_{\min}(\mathbf{B_p}(\alpha_p))}{\lambda_{\max}(\mathbf{A_p})} V_1(\mathbf{z}).
\end{equation} 
Consequently, by \cite[Theorem 4.10]{khalil_2002}, the origin $\mathbf{z} = \boldsymbol{0}$ of \eqref{eq:TransformedPositionDynamics} is exponentially stable for any initial condition originating in the arbitrarily large compact set $\boldsymbol{\Omega}_{\mathbf{p}_0}$. Given that the coordinate transformation $\mathbf{z} = \boldsymbol{\Phi}_{\mathbf{p}} \mathbf{\Tilde{x}_p}$ is a linear isomorphism, and therefore a global diffeomorphism mapping the origin uniquely to itself, the stability properties hold invariant under the transformation. As a result, the origin $\mathbf{\Tilde{x}_p} = \boldsymbol{0}$ of \eqref{eq:PositionTrackingDynamics} is globally asymptotically stable and semi-globally exponentially stable. 
\par Initializing the cascaded filter dynamics at $\mathbf{x_f}(0) \in \mathbb{B}^6_{\infty}$ ensures that $\mathbf{x_f}$, the uniform input bound $\|\mathbf{\bar{u}_p}\| \leq \sqrt{3}M_p$ guarantees $\|\mathbf{u_f}(t)\| \leq \sqrt{3}M_p$ for all $t \geq 0$. In light of \hyperref[def:SaturationFunction]{Definition~\ref*{def:SaturationFunction}} and recalling $T = m\|\mathbf{u}\|$, the triangle inequality applied to \eqref{eq:PositionControlLaw} establishes the upper bound
\begin{equation*}
	T \leq m \Big( \|\mathbf{u_f}\| + \|\mathbf{\ddot{p}}_d + g\mathbf{e_3}\| \Big) \leq T_{\max}.
\end{equation*}
Conversely, applying the reverse triangle inequality to $\mathbf{e_3^{\!\top}}\mathbf{u_p}$, and considering the strict bound $g - K_{a_3} > M$, yields
\begin{equation*}
	\|\mathbf{u}\| \geq |\mathbf{e}_3^\top \mathbf{u}| \geq g - |\mathbf{e}_3^\top \mathbf{u_f}| - |\mathbf{e}_3^\top \mathbf{\ddot{p}}_d| > g - M - K_{a_3} > 0. 
\end{equation*}
Consequently, $T \geq T_{\min} > 0$. Therefore, \eqref{eq:BoundsT} is formally satisfied for $\mathbf{r} \in \boldsymbol{\Omega}_{\mathbf{r}}$ along any solution $\mathbf{\Tilde{x}}_p$ to \eqref{eq:PositionTrackingDynamics} for all $t \geq 0$.
\end{proof}

\begin{figure}[t] 
    \centering
    
    \begin{subfigure}[t]{0.48\columnwidth} 
        \centering
        \begin{tikzpicture}[
            baseline=(int_phys.north), 
            scale=1, transform shape, 
            >=latex, 
            thick,
            block/.style={draw, rectangle, minimum height=0.7cm, minimum width=2cm, align=center}
        ]
            \node[block] (int_phys) at (0,0) {$\mathbf{F_p}(\mathbf{\Tilde{v}}, \mathbf{u_f})$};
            
            \node[block] (filt) at (0,-1.25) {$\mathbf{F_f}(\mathbf{x_f}, \mathbf{\bar{u}_p})$};

            \draw[->] (int_phys.east) -- ++(0.75,0) |- node[near end, above] {$\mathbf{\bar{u}_p}$} (filt.east);
            
            \draw[->] (filt.west) -- ++(-0.75,0) |- node[near end, above] {$\mathbf{u_f}$} (int_phys.west);

            \path (0, -1.8);
        \end{tikzpicture}
        \caption{Filtered System.}
        \label{fig:phys_loop}
    \end{subfigure}
    \hfill
    \begin{subfigure}[t]{0.48\columnwidth} 
        \centering
        \begin{tikzpicture}[
            baseline=(int_z.north), 
            scale=1, transform shape,
            >=latex, 
            thick,
            block/.style={draw, rectangle, minimum height=0.7cm, minimum width=2.4cm, align=center}
        ]
            \node[block] (int_z) at (0,0) {$\mathbf{F_z}(\mathbf{z_1}, \mathbf{z_2}, \mathbf{\bar{u}_p})$};

            \draw[->] (int_z.east) -- ++(0.625,0) |- (0,-1.25) -| ([xshift=-0.75cm]int_z.west) -- node[above] {$\mathbf{\bar{u}_p}$} (int_z.west);

            \path (0, -1.8); 
        \end{tikzpicture}
        \caption{Transformed System.}
        \label{fig:z_loop}
    \end{subfigure}

    \caption{Simplified structural comparison of the control loops. (a) The real system inherently embeds the finite-bandwidth filter dynamics within the feedback loop. (b) The global diffeomorphism maps the system to the transformed coordinates, algebraically masking the filter states.}
    \label{fig:PhysicalTransformedSystem}
\end{figure}
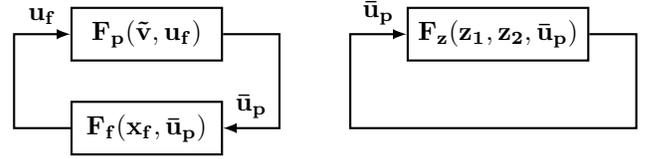

\section{Saturated Global Trajectory Tracking}
\label{section5}

\par The position tracking system dictates the normalized thrust direction $\boldsymbol{\rho} \coloneqq \frac{\mathbf{u}}{\|\mathbf{u}\|}$. By design, the bounded feedback law enforces $\mathbf{e_3^{\!\top}}\mathbf{u} > 0$, which inherently guarantees $\|\mathbf{u}\| > 0$ and the well-posedness of $\mathbf{R_d}$. Given the bounds derived in the preceding section, the forthcoming development assumes, without loss of generality, $\mathbf{x_f} \in M_p\mathbb{B}^6_\infty$. The desired angular velocity $\boldsymbol{\omega}_{\mathbf{d}} \in \mathbb{R}^3$ is extracted directly from the kinematics $\left[ \boldsymbol{\omega}_{\mathbf{d}} \right]_{\times} = \mathbf{R_d^{\!\top}}\mathbf{\dot{R}_d}$, yielding
\begin{equation}
\label{eq:DesiredAngularVelocity}
    \boldsymbol{\omega}_{\mathbf{d}} = \begin{bmatrix}
        \boldsymbol{\upsilon^{\!\top}} \left[\boldsymbol{\rho}\right]_\times \boldsymbol{\dot{\rho}} \\
        \boldsymbol{\upsilon^{\!\top}}\boldsymbol{\dot{\rho}} \\
        - \boldsymbol{\upsilon^{\!\top}} \left[\boldsymbol{\rho}\right]_\times \boldsymbol{\dot{\upsilon}}
    \end{bmatrix}.
\end{equation}
The desired angular acceleration $\boldsymbol{\dot{\omega}}_{\mathbf{d}} \in \mathbb{R}^3$ is subsequently obtained by differentiating with respect to time the previous expression, leading to
\begin{equation}
\label{eq:DesiredAngularAcceleration}
    \boldsymbol{\dot{\omega}}_{\mathbf{d}} = \begin{bmatrix}
     (\mathbf{e_2^{\!\top}}\boldsymbol{\omega}_{\mathbf{d}})
     (\mathbf{e_3^{\!\top}}\boldsymbol{\omega}_{\mathbf{d}}) \\
     -(\mathbf{e_1^{\!\top}}\boldsymbol{\omega}_{\mathbf{d}})(\mathbf{e_3^{\!\top}}\boldsymbol{\omega}_{\mathbf{d}})\\
     -(\mathbf{e_1^{\!\top}}\boldsymbol{\omega}_{\mathbf{d}})(\mathbf{e_2^{\!\top}}\boldsymbol{\omega}_{\mathbf{d}})
    \end{bmatrix}
    +
    \begin{bmatrix}
        \boldsymbol{\upsilon^{\!\top}} \left[\boldsymbol{\rho}\right]_\times \boldsymbol{\ddot{\rho}} \\
        \boldsymbol{\upsilon^{\!\top}}\boldsymbol{\ddot{\rho}} \\
        - \boldsymbol{\upsilon^{\!\top}} \left[\boldsymbol{\rho}\right]_\times \boldsymbol{\ddot{\upsilon}}
    \end{bmatrix}.
\end{equation}

\hyperref[lem:AttitudeReferencesBounds]{Lemma~\ref*{lem:AttitudeReferencesBounds}} demonstrates the uniform boundedness of the resulting attitude references $\boldsymbol{\omega}_{\mathbf{d}}$ and $\boldsymbol{\dot{\omega}}_{\mathbf{d}}$.

\begin{lem}
\label{lem:AttitudeReferencesBounds}
Suppose the conditions of \hyperref[assump:ReferenceTrajectory]{Assumption~\ref*{assump:ReferenceTrajectory}} hold for all $t \geq 0$ and the filter states satisfy the initial condition $\mathbf{x_f}(0) \in M_p\mathbb{B}^6_\infty$. Then, there exist finite positive constants $U_{\omega_d}$ and $U_{\dot{\omega}_d}$ such that $\|\boldsymbol{\omega}_{\mathbf{d}}(t)\| \leq U_{\omega_d}$ and $\|\boldsymbol{\dot{\omega}}_{\mathbf{d}}(t)\| \leq U_{\dot{\omega}_d}$ for all $t \geq 0$.
\end{lem}
\begin{proof}
    The theoretical upper bounds are established by bounding the thrust direction derivatives, $\boldsymbol{\dot{\rho}}$ and $\boldsymbol{\ddot{\rho}}$, alongside the derivatives of $\mathbf{R_d}\mathbf{e_1}$, $\dot{\boldsymbol{\upsilon}}$ and $\ddot{\boldsymbol{\upsilon}}$. The former pair yields
    \begin{equation*}
        \boldsymbol{\dot{\rho}} = \frac{1}{\|\mathbf{u}\|} (\mathbf{I} - \boldsymbol{\rho}\boldsymbol{\rho^{\!\top}}) \mathbf{\dot{u}} \leq \frac{m\|\mathbf{\dot{u}}\|}{T_{\min}},
    \end{equation*}
    \begin{equation*}
        \boldsymbol{\ddot{\rho}} = \frac{1}{\|\mathbf{u}\|} (\mathbf{I} - \boldsymbol{\rho}\boldsymbol{\rho^{\!\top}}) \mathbf{\ddot{u}} - \frac{2(\boldsymbol{\rho^{\!\top}}\mathbf{\dot{u}})}{\|\mathbf{u}\|}\boldsymbol{\dot{\rho}} - \|\boldsymbol{\dot{\rho}}\|^2\boldsymbol{\rho}.
    \end{equation*}
Given the filter dynamics \eqref{eq:FilterDynamics} and \hyperref[assump:ReferenceTrajectory]{Assumption~\ref*{assump:ReferenceTrajectory}}, 
\begin{equation*}
    \|\mathbf{\dot{u}}\| \leq k_f\|\mathbf{u_f} - \mathbf{u_s}\| + K_j,
\end{equation*}
\begin{equation*}
    \|\mathbf{\ddot{u}}\| \leq k_f(k_f\|\mathbf{u_f} - \mathbf{u_s}\| + k_s\|\mathbf{u_s} - \mathbf{\bar{u}_p}\|) + K_s.
\end{equation*}
The input satisfies $\|\mathbf{\bar{u}_p}\| \leq \sqrt{3}M_p$. Thus, in light of the linear filter dynamics and for the initial condition $\mathbf{x_f}(0) = \boldsymbol{0}$, the absolute maximum of the error $\mathbf{u_f} - \mathbf{u_s}$ can be computed by applying the $\mathcal{L_1}$-norm to the impulse response of the transfer function of this filter states difference:
\begin{equation*}
    \max \|\mathbf{u_f} - \mathbf{u_s}\| \leq \sqrt{3}M_p\int_0^\infty |h_f(t)|\textrm{d} t
\end{equation*}
with the impulse response $h_f \colon \mathbb{R}_{\geq 0} \to \mathbb{R}$ satisfying
\begin{equation*}
    h_f(t) = 
\begin{cases} 
\frac{k_s}{k_s - k_f} \left( k_s e^{-k_s t} - k_f e^{-k_f t} \right), & \text{if } k_s \neq k_f \\ 
k_f (1 - k_f t) e^{-k_f t}, & \text{if } k_s = k_f 
\end{cases}
\end{equation*}
This integrand is equal to zero exactly once at $t = t^*$, with
\begin{equation*}
    t^* = 
\begin{cases} 
\frac{\ln(k_s / k_f)}{k_s - k_f}, & \text{if } k_s \neq k_f \\ 
k_f^{-1}, & \text{if } k_s = k_f 
\end{cases}.
\end{equation*}
By partitioning the absolute value integral at $t^*$, the $\mathcal{L}_1$-norm yields exactly twice the value of the step response evaluated at $t = t^*$: 
\begin{equation*}
    \int_{0}^{\infty} \!\!|h_f(t)| \textrm{ d}t \!=\! 2 U_{l}, \; \textrm{with} \; U_{l} \!=\!  
\begin{cases} 
\left( \frac{k_f}{k_s} \right)^{\frac{k_f}{k_s - k_f}}, & \text{if } k_s \neq k_f \\ 
e^{-1}, & \text{if } k_s = k_f 
\end{cases}.
\end{equation*}
For an arbitrary initial condition $\mathbf{x_f}(0) \in M_p\mathbb{B}^6_\infty$, it directly follows that $\|\mathbf{u_f} - \mathbf{u_s}\| \leq 2\sqrt{3}M_p$, corresponding to $U_l = 1$. Consequently, considering that $\|\mathbf{u_s} - \mathbf{\bar{u}_p}\| \leq 2\sqrt{3}M_p$ for $\mathbf{x_f}(0) = \boldsymbol{0}$,
\begin{equation*}
    \|\mathbf{\dot{u}}\| \leq 2 \sqrt{3} k_f U_l M_p + K_j = U_{\dot{u}},
\end{equation*}
\begin{equation*}
    \|\mathbf{\ddot{u}}\| \leq 2 \sqrt{3}k_fM_p(k_fU_l + k_s) + K_s = U_{\ddot{u}}.
\end{equation*}
Based on these latter bounds, one has 
\begin{equation*}
    \|\boldsymbol{\dot{\rho}}\| \leq \frac{mU_{\dot{u}}}{T_{\min}} = U_{\dot{\rho}}, \quad \|\boldsymbol{\ddot{\rho}}\| \leq \frac{mU_{\ddot{u}}}{T_{\min}} + 3 U^2_{\dot{\rho}} = U_{\ddot{\rho}}.
\end{equation*}
Focusing on the time derivatives of $\boldsymbol{\upsilon}$, note that
    \begin{equation*}
        \boldsymbol{\dot{\upsilon}} = \frac{1}{\|\boldsymbol{\varpi}\|} (\mathbf{I} - \boldsymbol{\upsilon}\boldsymbol{\upsilon^{\!\top}}) \boldsymbol{\dot{\varpi}},
    \end{equation*}
    \begin{equation*}
        \boldsymbol{\ddot{\upsilon}} = \frac{1}{\|\boldsymbol{\varpi}\|} (\mathbf{I} - \boldsymbol{\upsilon}\boldsymbol{\upsilon^{\!\top}}) \boldsymbol{\ddot{\varpi}} - \frac{2(\boldsymbol{\upsilon^{\!\top}}\boldsymbol{\dot{\varpi}})}{\|\boldsymbol{\varpi}\|}\boldsymbol{\dot{\upsilon}} - \|\boldsymbol{\dot{\upsilon}}\|^2\boldsymbol{\upsilon}.
    \end{equation*}
Considering \hyperref[assump:ReferenceTrajectory]{Assumption~\ref*{assump:ReferenceTrajectory}}, the time derivatives of $\boldsymbol{\varpi}$ satisfy
\begin{equation*}
    \|\boldsymbol{\dot{\varpi}}\| \leq \|\boldsymbol{\dot{\rho}}\| + \|\boldsymbol{\dot{\nu}}_{\mathbf{d}}\| \leq U_{\dot{\rho}} + K_{\dot{\nu}} = U_{\dot{\varpi}},
\end{equation*}
\begin{equation*}
    \|\boldsymbol{\ddot{\varpi}}\| \leq \|\boldsymbol{\ddot{\rho}}\| + 2 \|\boldsymbol{\dot{\rho}}\|\|\boldsymbol{\dot{\nu}}_{\mathbf{d}}\| + \|\boldsymbol{\ddot{\nu}}_{\mathbf{d}}\| \leq U_{\ddot{\rho}} + 2U_{\dot{\rho}}K_{\dot{\nu}} + K_{\ddot{\nu}} = U_{\ddot{\varpi}}.
\end{equation*}
To establish a rigorous, non-conservative lower bound on $\|\boldsymbol{\varpi}\|$, the heading projection must be evaluated at the exact boundaries of the control allocation envelope. The magnitude of this vector reaches its geometric minimum when the vertical component of the thrust direction is minimized, i.e, $\min \|\boldsymbol{\varpi}\| = \min |\mathbf{e_3^{\!\top}}\boldsymbol{\rho}|$. Since the ratio $\mathbf{e_3^{\!\top}}\boldsymbol{\rho} = \frac{|\mathbf{e_3^{\!\top}}\mathbf{u}|}{\|\mathbf{u}\|}$ monotonically decreases as lateral thrust increases and monotonically increases as vertical thrust increases, the strict global infimum results from simultaneously minimizing the vertical control effort while maximizing the orthogonal lateral efforts. Furthermore, applying the Cauchy-Schwarz inequality to the lateral components leads to
\begin{equation*}
    (\mathbf{e_1^{\!\top}}\mathbf{u})^2 + (\mathbf{e_2^{\!\top}}\mathbf{u})^2 \leq 2M_p^2 +K_{a_{1,2}}^2 + 2\sqrt{2}M_pK_{a_{1,2}}.
\end{equation*}
Consequently, the lower bound is explicitly formulated as:
\begin{equation*}
\|\boldsymbol{\varpi}\| \geq \frac{T_{\min}}{m \sqrt{\left(\frac{T_{\min}}{m}\right)^2 + 2M_p^2 +K_{a_{1,2}}^2 + 2\sqrt{2}M_pK_{a_{1,2}}}} = L_\varpi.
\end{equation*}
With these bounds in place, the derivatives of $\boldsymbol{\upsilon}$ are upper-bounded by
\begin{equation*}
    \|\boldsymbol{\dot{\upsilon}}\| \leq \frac{U_{\dot{\varpi}}}{L_{\varpi}} = U_{\dot{\upsilon}}, \quad \|\boldsymbol{\ddot{\upsilon}}\| \leq \frac{U_{\ddot{\varpi}}}{L_{\varpi}} + 3U_{\dot{\upsilon}}^2 = U_{\ddot{\upsilon}}.
\end{equation*}
\par In light of $\boldsymbol{\rho}$ being a unit vector, the derivative $\boldsymbol{\dot{\rho}}$ is orthogonal to the desired thrust direction. As a result, $\boldsymbol{\dot{\rho}}$ lies in the plane spanned by $\boldsymbol{\upsilon}$ and $\left[\boldsymbol{\rho}\right]_{\times}\boldsymbol{\upsilon}$. Therefore, $\|\boldsymbol{\dot{\rho}}\|^2$ is just the sum of the squares of its projection onto those two basis vectors. Given that $\boldsymbol{\upsilon}$ is a unit vector as well, a similar logic applies to $\boldsymbol{\dot{\upsilon}}$. Then, recalling \eqref{eq:DesiredAngularVelocity} and noticing that, due to orthogonality, $\boldsymbol{\upsilon^{\!\top}} \boldsymbol{\dot{\rho}} = - \boldsymbol{\rho^{\!\top}}\boldsymbol{\dot{\upsilon}}$, 
\begin{equation*}
    (\mathbf{e_1^{\!\top}}\boldsymbol{\omega}_{\mathbf{d}})^2 + (\mathbf{e_2^{\!\top}}\boldsymbol{\omega}_{\mathbf{d}})^2 = \|\boldsymbol{\dot{\rho}}\|^2, \quad (\mathbf{e_2^{\!\top}}\boldsymbol{\omega}_{\mathbf{d}})^2 + (\mathbf{e_3^{\!\top}}\boldsymbol{\omega}_{\mathbf{d}})^2 = \|\boldsymbol{\dot{\upsilon}}\|^2,
\end{equation*}
leading to 
\begin{equation}
\label{eq:FirstTermAngularAccelerationDesired}
\begin{aligned}
    \|\boldsymbol{\dot{\rho}}\|^2 \|\boldsymbol{\dot{\upsilon}}\|^2 = & (\mathbf{e}_2^{\!\top}\boldsymbol{\omega}_{\mathbf{d}})^2 (\mathbf{e}_3^{\!\top}\boldsymbol{\omega}_{\mathbf{d}})^2 + (\mathbf{e}_1^{\!\top}\boldsymbol{\omega}_{\mathbf{d}})^2 (\mathbf{e}_3^{\!\top}\boldsymbol{\omega}_{\mathbf{d}})^2 \\ & + (\mathbf{e}_1^{\!\top}\boldsymbol{\omega}_{\mathbf{d}})^2 (\mathbf{e}_2^{\!\top}\boldsymbol{\omega}_{\mathbf{d}})^2 + (\mathbf{e}_2^{\!\top}\boldsymbol{\omega}_{\mathbf{d}})^4,
\end{aligned}
\end{equation}
and the upper-bound 
\begin{equation*}
    \|\boldsymbol{\omega}_{\mathbf{d}}\| \leq \sqrt{U_{\dot{\rho}}^2 + U^2_{\dot{\upsilon}}} = U_{\omega_d}.
\end{equation*}
Unlike $\boldsymbol{\dot{\rho}}$, the derivative $\boldsymbol{\ddot{\rho}}$ is not orthogonal to $\boldsymbol{\rho}$. Consequently,
\begin{equation*}
    (\boldsymbol{\upsilon^{\!\top}} \left[\boldsymbol{\rho}\right]_\times \boldsymbol{\ddot{\rho}})^2 + (
        \boldsymbol{\upsilon^{\!\top}}\boldsymbol{\ddot{\rho}})^2 \leq \|\boldsymbol{\ddot{\rho}}\|^2.
\end{equation*}
Evaluating \eqref{eq:DesiredAngularAcceleration} by applying \eqref{eq:FirstTermAngularAccelerationDesired} to the first term and leveraging the previously established bounds yields
\begin{equation*}
    \|\boldsymbol{\dot{\omega}}_{\mathbf{d}}\| \!\leq\! \|\boldsymbol{\dot{\rho}}\|\|\boldsymbol{\dot{\upsilon}}\| + \!\sqrt{\|\boldsymbol{\ddot{\rho}}\|^2 \!+\! \|\boldsymbol{\ddot{\upsilon}}\|^2} \!\leq\! U_{\dot{\rho}}U_{\dot{\upsilon}} \!+\! \sqrt{U_{\ddot{\rho}}^2 \!+\! U_{\ddot{\upsilon}}^2} = U_{\dot{\omega}_d}.
\end{equation*}
This establishes the uniform upper bounds $U_{\omega_d}$ and $U_{\dot{\omega}_d}$, concluding the proof.
\end{proof}
\begin{rem}
The introduced filter dynamics lead to structurally simplified expressions for the derivatives of the feedback law $\mathbf{u}$. As a result, compared to \cite[Lemma 5]{martins2024}, the proposed framework yields analytically tractable upper bounds on $\boldsymbol{\omega}_{\mathbf{d}}$ and $\boldsymbol{\dot{\omega}}_{\mathbf{d}}$ that are systematically tunable via the control parameters.
\hfill $\square$
\end{rem}
\par Consider the rotation matrix error $\mathbf{\Tilde{R}} \coloneqq \mathbf{R_d^{\!\top}}\mathbf{R} \in \mathrm{SO}(3)$. Leveraging the framework of \cite{martins2025hybrid, martins2024CDCattitude}, the hybrid dynamic mechanism $\mathcal{H}_{\mathbf{l}}$ detailed in \hyperref[appendix:PathLiftingAlgorithm]{Appendix~\ref*{appendix:PathLiftingAlgorithm}} uniquely and robustly lifts the attitude trajectory, yielding an MRP error $\boldsymbol{\Tilde{\vartheta}} \in (1+\delta)\mathbb{B}^3$ that satisfies $\mathcal{R}_{\boldsymbol{\vartheta}}(\boldsymbol{\Tilde{\vartheta}}) = \mathbf{\Tilde{R}}$. Let the lifted attitude error state be defined as $\mathbf{\Tilde{x}}_{\boldsymbol{\vartheta}} \coloneqq (\boldsymbol{\Tilde{\vartheta}}, \boldsymbol{\Tilde{\omega}}) \in \boldsymbol{\chi_{\Tilde{\vartheta}}}$, with the domain $\boldsymbol{\chi_{\Tilde{\vartheta}}} \coloneqq (1 + \delta)\mathbb{B}^3 \times \mathbb{R}^3$, where the angular velocity error is given by $\boldsymbol{\Tilde{\omega}} \coloneqq \boldsymbol{\omega} - \mathbf{\Tilde{R}}^{\!\top}\boldsymbol{\omega}_{\mathbf{d}}$. To enforce the actuation limits, the following saturated feedback control law is proposed:
\begin{equation*}
    \boldsymbol{\tau}(\mathbf{\tilde{x}_{\boldsymbol{\vartheta}}}) \coloneqq -\frac{1 + \|\boldsymbol{\tilde{\vartheta}}\|^2}{4}\boldsymbol{s_{\vartheta}}(k_\vartheta\boldsymbol{\tilde{\vartheta}}) - \boldsymbol{s_{\omega}}(k_\omega\boldsymbol{\tilde{\omega}}) -\boldsymbol{\tau_c},
\end{equation*}
where $\boldsymbol{s_{\vartheta}}, \boldsymbol{s_{\omega}}\colon \mathbb{R}^3 \to \mathbb{R}^3$ denote radial saturation functions bounded by the levels $M_\vartheta, M_\omega \in\mathbb{R}_{>0}$ given by
\begin{equation*}
    \boldsymbol{s_{\vartheta}}(\boldsymbol{\Tilde{\vartheta}}) \coloneqq \frac{M_\vartheta k_\vartheta \boldsymbol{\Tilde{\vartheta}}}{\sqrt{M_\vartheta^2 + \|k_\vartheta \boldsymbol{\Tilde{\vartheta}}\|^2}}, \quad \boldsymbol{s_{\omega}}(\boldsymbol{\Tilde{\omega}}) \coloneqq \frac{M_\omega k_\omega \boldsymbol{\Tilde{\omega}}}{\sqrt{M_\omega^2 + \|k_\omega \boldsymbol{\Tilde{\omega}}\|^2}}.  
\end{equation*}
The constants $k_\vartheta, k_\omega \in \mathbb{R}_{>0}$ dictate the proportional and derivative feedback gains. Furthermore, $\boldsymbol{\tau_c} \in \mathbb{R}^3$ serves as a feedforward canceling term defined as
\begin{equation*}
    \boldsymbol{\tau_c} \coloneqq [\mathbf{J}\mathbf{\tilde{R}^{\!\top}}\boldsymbol{\omega_d}]_\times \mathbf{\tilde{R}^{\!\top}}\boldsymbol{\omega_d} - \mathbf{J}\mathbf{\tilde{R}^{\!\top}}\boldsymbol{\dot{\omega}_d}.
\end{equation*}
Define the state-space $\boldsymbol{\chi_\vartheta} \!=\! \mathrm{SO(3)} \times U_{\dot{\omega}_d} \mathbb{B}^3 \times \boldsymbol{\chi_{\Tilde{\vartheta}}}$ and the state-vector $\mathbf{x}_{\boldsymbol{\vartheta}}  \coloneqq  (\mathbf{R_d},\boldsymbol{\omega}_{\mathbf{d}},\mathbf{\Tilde{x}}_{\boldsymbol{\vartheta}}) \in \boldsymbol{\chi_\vartheta}$. The quadruplet
\vspace{-0.1500cm}
\begin{subequations}
\begin{equation*}
    \mathbf{C}_{\boldsymbol{\vartheta}} \!\coloneqq\! \{ \mathbf{x}_{\boldsymbol{\vartheta}} \!\in\! \boldsymbol{\chi_{\vartheta}}\!\colon \!\|\boldsymbol{\Tilde{\vartheta}}\| \!\le\! 1 + \delta \}, \quad \!\! \mathbf{F}_{\boldsymbol{\vartheta}}(\mathbf{x}_{\boldsymbol{\vartheta}}) \!\coloneqq\! \left(\begin{array}{c}
        \mathbf{R_d}\left[\boldsymbol{\omega}_{\mathbf{d}}\right]_{\times} \\
        K_{\dot{\omega}_{d}}\mathbb{B}^3 \\
        \mathbf{F}_{\boldsymbol{\Tilde{\vartheta}}} (\mathbf{x}_{\boldsymbol{\vartheta}})
        \end{array}\!\!\right) 
\vspace{-0.1500cm}
\end{equation*}
\begin{equation*}
     \mathbf{D}_{\boldsymbol{\vartheta}} \!\coloneqq\! \{ \mathbf{x}_{\boldsymbol{\vartheta}} \!\in\! \boldsymbol{\chi_{\vartheta}}\!\colon \!\|\boldsymbol{\Tilde{\vartheta}}\| \!=\! 1 + \delta \}, \quad \!\!\!\!\mathbf{G}_{\boldsymbol{\vartheta}}(\mathbf{x}_{\boldsymbol{\vartheta}}) \!\coloneqq\! (\mathbf{R_d}, \boldsymbol{\omega}_{\mathbf{d}}, \!\boldsymbol{\Upsilon}(\boldsymbol{\Tilde{\vartheta}}), \boldsymbol{\Tilde{\omega}})
\vspace{-0.1500cm}
\end{equation*}
\end{subequations}
with
\begin{equation}
    \mathbf{F}_{\boldsymbol{\Tilde{\vartheta}}} (\mathbf{x}_{\boldsymbol{\vartheta}}) \coloneqq \left(\begin{array}{c}
        \mathbf{T}(\boldsymbol{\Tilde{\vartheta}})\boldsymbol{\Tilde{\omega}} \\
        \mathbf{J}^{-1}\left(\boldsymbol{\Delta}\left(\boldsymbol{\Tilde{\omega}}, \mathbf{\Tilde{R}^\top}\boldsymbol{\omega}_\mathbf{d} \right) + \boldsymbol{\tau} + \boldsymbol{\tau_c}\right)
        \end{array}\!\!\right), 
\end{equation}
where $\boldsymbol{\Delta}\left(\boldsymbol{\Tilde{\omega}}, \mathbf{\Tilde{R}^\top}\boldsymbol{\omega}_\mathbf{d} \right)$ satisfies $\boldsymbol{\Tilde{\omega}^{\!\top}}\boldsymbol{\Delta}\left(\boldsymbol{\Tilde{\omega}}, \mathbf{\Tilde{R}^\top}\boldsymbol{\omega}_\mathbf{d} \right) = 0$,
defines the closed-loop hybrid lifted attitude tracking system $\mathcal{H}_{\boldsymbol{\vartheta}} = (\mathbf{C}_{\boldsymbol{\vartheta}}, \mathbf{F}_{\boldsymbol{\vartheta}}, \mathbf{D}_{\boldsymbol{\vartheta}}, \mathbf{G}_{\boldsymbol{\vartheta}})$. The switching logic between the principal and shadow MRP coordinates incorporates a hysteresis margin $\delta \in \mathbb{R}_{>0}$. Sizing $\delta$ to upper-bound measurement noise eliminates chattering and guarantees robust transitions.

\begin{thm}
\label{thm:StabilityAttitudeTrackingSystemMRP}
\label{chapter4thm:InnerLoopStability}
	Let \hyperref[assump:ReferenceTrajectory]{Assumption~\ref*{assump:ReferenceTrajectory}} hold for all $t \geq 0$. The hybrid system $\mathcal{H}_{\boldsymbol{\vartheta}}$ is well-posed and the compact set $\mathcal{A}_{\boldsymbol{\vartheta}} = \left\{\mathbf{x}_{\boldsymbol{\vartheta}} \in \boldsymbol{\chi_{\vartheta}} \colon \mathbf{\Tilde{x}_{\boldsymbol{\vartheta}}} = \boldsymbol{0}\right\}$ is globally asymptotically stable and semi-globally exponentially stable for $\mathcal{H}_{\boldsymbol{\vartheta}}$. Furthermore, let the design parameters satisfy
\begin{equation}
\label{eq:TorqueMaxConditionsTheorem}
    \mathbf{e_i^{\!\top}}\boldsymbol{\tau_{\max}} > J_{ii} U_{\dot{\omega}_d} + \tfrac{1}{2} \left| J_{jj} - J_{kk} \right| U_{\omega_d}^2 + \iota,
\end{equation}
with
\begin{equation*}
		\iota = (1 + (1+\delta)^2)\tfrac{M_\vartheta}{4} + M_\omega.
\end{equation*}
for all $i \in \{1,2,3\}$, where $j$ and $k$ are the remaining mutually distinct indices such that $\{j, k\} = \{1, 2, 3\} \setminus \{i\}$. Then , $\boldsymbol{\tau} \in \boldsymbol{\Omega_\tau}$ for any solution $\mathbf{x}_{\boldsymbol{\vartheta}}(t,j)$ to $\mathcal{H}_{\boldsymbol{\vartheta}}$.
\end{thm}


\begin{proof}
    Using arguments analogous to those presented in \cite[Lemma 6]{martins2024}, $\mathcal{H}_{\boldsymbol{\vartheta}}$ satisfies the hybrid basic conditions \cite[Assumption 6.5]{goebel_2012}, guaranteeing its well-posedness \cite[Theorem 6.30]{goebel_2012}. Define the Lyapunov function $V_2 \colon \boldsymbol{\chi_{\vartheta}} \to \mathbb{R}_{\geq 0}$
    \begin{equation*}
        V_2(\mathbf{x}_{\boldsymbol{\vartheta}}) \!\coloneqq \!a\boldsymbol{\Tilde{\omega}^{\!\top}}\mathbf{J}\boldsymbol{\Tilde{\omega}} +  \frac{2a}{k_{\vartheta}} \int_{\boldsymbol{0}}^{k_\vartheta\boldsymbol{\Tilde{\vartheta}}} \!\!\!\!\!\!\!\!\!\! \boldsymbol{s_{\vartheta}}\!(\boldsymbol{\mu})^{\!\top} \!\; \textrm{d}\boldsymbol{\mu} \, + \, b \boldsymbol{\Tilde{\vartheta}^{\!\top}}\boldsymbol{\sigma_{\omega}}(k_\omega \mathbf{J}\boldsymbol{\Tilde{\omega}}) ,
    \end{equation*}
with $\boldsymbol{\sigma_{\omega}}\colon \mathbb{R}^3 \to \mathbb{R}^3$ denotes an auxiliary saturation function bounded by $M_\omega^* = M_\omega\lambda_{\min}(\mathbf{J})$, and $a, b \in \mathbb{R}_{>0}$ satisfying
\begin{subequations}
\begin{equation}
\label{eq:ConditionLyapunovGain1}
a > \max\left\{2\frac{k_\omega^2}{\gamma_\vartheta}, \beta_\vartheta\lambda_{\max}(\mathbf{J})\right\}b,
\end{equation}
\begin{equation}
\label{eq:ConditionLyapunovGain2}
    a > \max\left\{\frac{b}{2} \frac{k_\omega\lambda_{\max}(\mathbf{J})}{\sqrt{\gamma_\vartheta \lambda_{\min}(\mathbf{J})}},   \frac{\alpha_\vartheta}{\delta} \frac{M_\omega^*}{2k_\vartheta} b + \frac{1}{2k_\vartheta}\right\},
\end{equation}
\end{subequations}
with $\alpha_\vartheta \coloneqq 1 + \delta$, $\beta_\vartheta \coloneqq (1 + \alpha_\vartheta^2)4^{-1}$, and $\gamma_\vartheta \coloneqq k_\vartheta M_\vartheta(M_\vartheta^2 + (k_\vartheta \alpha_\vartheta)^2)^{-1/2}$. The function $V_2$ is continuously differentiable on $\boldsymbol{\chi_{\vartheta}}$ and radially unbounded. Consequently, for any given initial condition $\mathbf{\Tilde{x}}_{\boldsymbol{\vartheta}}(0,0)$, the sublevel set $\mathcal{V}_{\mathbf{2}} = \{\mathbf{\Tilde{x}}_{\boldsymbol{\vartheta}} \in \boldsymbol{\chi_\vartheta}: V_2(\mathbf{\Tilde{x}}_{\boldsymbol{\vartheta}}) \leq V_2(\mathbf{\Tilde{x}}_{\boldsymbol{\vartheta}}(0,0))\}$ is compact. Furthermore, $V_2$ satisfies the upper-bound $V_2 \leq \lambda_{\max}\left(\mathbf{A}_{\boldsymbol{\vartheta_1}}\right) \|\mathbf{\Tilde{x}}_{\boldsymbol{\vartheta}}\|^2$, with 
\vspace{-0.2cm}
\begin{equation*}
        \mathbf{A}_{\boldsymbol{\vartheta_1}} = \frac{1}{2}\left[\begin{array}{cc}
            2ak_\vartheta & bk_\omega\lambda_{\max}(\mathbf{J}) \\
            bk_\omega\lambda_{\max}(\mathbf{J}) & 2a\lambda_{\max}(\mathbf{J})
        \end{array}\right]
\vspace{-0.2cm}
\end{equation*}
Since the hybrid path-lifting algorithm strictly ensures $\|\boldsymbol{\Tilde{\vartheta}}\| \leq \alpha_\vartheta$, the bound $\|\boldsymbol{s_\vartheta}(k_\vartheta \boldsymbol{\Tilde{\vartheta}})\| \geq \gamma_\vartheta \|\boldsymbol{\Tilde{\vartheta}}\|$ holds for all $(t,j) \in \textrm{dom } \mathbf{x}_{\boldsymbol{\vartheta}}$. This establishes the quadratic lower bound $V_2 \geq \lambda_{\min} \left(\mathbf{A}_{\boldsymbol{\vartheta_2}} \right) \|\mathbf{\Tilde{x}}_{\boldsymbol{\vartheta}}\|^2$, where
\vspace{-0.2cm}
\begin{equation*}
        \mathbf{A}_{\boldsymbol{\vartheta_2}} = \frac{1}{2}\left[\begin{array}{cc}
            2a\gamma_\vartheta & -bk_\omega\lambda_{\max}(\mathbf{J}) \\
            -bk_\omega\lambda_{\max}(\mathbf{J}) & 2a\lambda_{\min}(\mathbf{J})
        \end{array}\right].
\vspace{-0.2cm}
\end{equation*}
By virtue of \eqref{eq:ConditionLyapunovGain2}, the matrix $\mathbf{A}_{\boldsymbol{\vartheta_2}}$ is positive definite, confirming that $V_2$ is positive definite with respect to $\mathcal{A}_{\boldsymbol{\vartheta}}$.
\par The saturation level $M_\omega^*$ allows writing $\|\boldsymbol{\sigma_{\omega}}(k_\omega\mathbf{J}\boldsymbol{\Tilde{\omega}})\| \leq \lambda_{\max} (\mathbf{J})\| \boldsymbol{s_\omega}(k_\omega \boldsymbol{\Tilde{\omega}})\|$. Furthermore, based on \cite[p. 123]{junkins_2009}, 
\begin{equation*}
    \boldsymbol{\Tilde{\vartheta}^{\!\top}} \mathbf{T}(\boldsymbol{\Tilde{\vartheta}}) = \boldsymbol{\Tilde{\vartheta}^{\!\top}} \frac{1 + \|\boldsymbol{\Tilde{\vartheta}}\|^2}{4},
\end{equation*}
it follows that 
\begin{equation*}
    \boldsymbol{s_\vartheta}( k_{\vartheta} \boldsymbol{\Tilde{\vartheta}})^{\!\top}\mathbf{T}(\boldsymbol{\Tilde{\vartheta}}) = \boldsymbol{s_\vartheta}( k_{\vartheta} \boldsymbol{\Tilde{\vartheta}})^{\!\top} \frac{1 + \|\boldsymbol{\Tilde{\vartheta}}\|^2}{4}.
\end{equation*}
Using these results and \eqref{eq:ConditionLyapunovGain1}, the derivative $\dot{V}_2(\mathbf{\Tilde{x}}_{\boldsymbol{\vartheta}})$ yields
\begin{equation*}
    \dot{V}_2(\mathbf{\Tilde{x}_{\boldsymbol{\vartheta}}}) \leq -\tfrac{c}{k_\omega}\|\boldsymbol{s_\omega}(k_\omega\boldsymbol{\Tilde{\omega}})\|^2 -bk_\omega\boldsymbol{\Tilde{\vartheta}^{\top}}\!\boldsymbol{\Lambda}(\boldsymbol{\Tilde{\omega}}) \!\left(\!\tfrac{\gamma_\vartheta\boldsymbol{\Tilde{\vartheta}}}{4} \!- \!\boldsymbol{s_\omega(k_\omega\boldsymbol{\Tilde{\omega}})}\!\right)
\end{equation*}
where $\boldsymbol{\Lambda}(\boldsymbol{\Tilde{\omega}}) = \frac{\partial \boldsymbol{\sigma_\omega}(k_\omega\boldsymbol{\Tilde{\omega}})}{\partial k_\omega\boldsymbol{\Tilde{\omega}}}$ is a diagonal matrix whose entries lie in the interval $(0, 1]$. Rewriting in quadratic form leads to
\begin{equation}
\label{eq:BoundV2dot}
    \dot{V}_2 \!\leq\! -\!\mathbf{z}_{\boldsymbol{\vartheta}}^{\!\top}
    \mathbf{B}_{\boldsymbol{\vartheta}}
    \mathbf{z}_{\boldsymbol{\vartheta}} - \frac{b}{8}k_\omega\gamma_\vartheta\boldsymbol{\Tilde{\vartheta}}^{\!\top}\boldsymbol{\Lambda}(\boldsymbol{\Tilde{\omega}}) \boldsymbol{\Tilde{\vartheta}},
\end{equation}
with 
\begin{equation*}
    \mathbf{z}_{\boldsymbol{\vartheta}}  \coloneq \begin{bmatrix}
       \!\boldsymbol{s_{\omega}}(k_\omega\boldsymbol{\Tilde{\omega}})\! \\ \boldsymbol{\Tilde{\vartheta}}
    \end{bmatrix}, \quad \mathbf{B}_{\boldsymbol{\vartheta}} \coloneqq \begin{bmatrix}
       \tfrac{c}{k_\omega}\mathbf{I_3 } & \tfrac{b}{2}\boldsymbol{\Lambda}(\boldsymbol{\Tilde{\omega}})k_\omega \\ 
       \tfrac{b}{2}\boldsymbol{\Lambda}(\boldsymbol{\Tilde{\omega}})k_\omega \!& \tfrac{b}{8}\boldsymbol{\Lambda}(\boldsymbol{\Tilde{\omega}})k_\omega\gamma_\vartheta
    \end{bmatrix}
\end{equation*}
By applying the Schur complement, condition \eqref{eq:ConditionLyapunovGain1} ensures that the central block matrix is strictly positive definite. Consequently, $\dot{V}_2$ is negative definite with respect to the compact set $\mathcal{A}_{\boldsymbol{\vartheta}}$. The discrete evolution of $V_2$ satisfies
\begin{equation*}
\begin{aligned}
V_2\!\left(\mathbf{G}_{\boldsymbol{\vartheta}}\left(\mathbf{x}_{\boldsymbol{\vartheta}}\right)\right) - & V_2\left(\mathbf{x}_{\boldsymbol{\vartheta}}\right) = \\ & \frac{2a}{k_\vartheta} \int^{k_\vartheta\boldsymbol{\Upsilon}(\boldsymbol{\Tilde{\vartheta}})}_{k_\vartheta \boldsymbol{\Tilde{\vartheta}}} \!\!\!\!\!\!\!\!\!\!\!\!\! \boldsymbol{s_\vartheta}(\boldsymbol{\mu}\!)^{\!\top} \!\textrm{d}\boldsymbol{\mu } - b \boldsymbol{\Tilde{\vartheta}}^{\!\top}\!\!\!\boldsymbol{\sigma_{\!\omega}}(k_\omega\mathbf{J}\mathbf{\boldsymbol{\Tilde{\omega}}}\!)\!\left(\!1 \!+\! \tfrac{1}{\|\boldsymbol{\Tilde{\vartheta}}\|^{2}}\!\!\right).	
\end{aligned}
\end{equation*}
Leveraging the previous bounds, it follows that
\begin{equation*}
    V_2(\!\mathbf{G}_{\boldsymbol{\vartheta}}(\mathbf{x}_{\boldsymbol{\vartheta}})\!) - V_2(\mathbf{x}_{\boldsymbol{\vartheta}}\!) \leq \!-a\!\left(\!\alpha_\vartheta^2 - \tfrac{1}{\alpha_\vartheta^2}\!\right)\!\gamma_\vartheta + bM_\omega^*\left(\!\alpha_\vartheta + \tfrac{1}{\alpha_\vartheta}\!\right)\!.
\end{equation*}
Then, applying condition \eqref{eq:ConditionLyapunovGain2} yields
\begin{equation}
    \label{eq:V2JumpBehavior}
    V_2(\!\mathbf{G}_{\boldsymbol{\vartheta}}(\mathbf{x}_{\boldsymbol{\vartheta}})) \leq e^{-\alpha_\delta} V_2(\mathbf{x}_{\boldsymbol{\vartheta}}) \quad \forall \quad \mathbf{\Tilde{x}}_{\boldsymbol{\vartheta}} \in \mathbf{D}_{\boldsymbol{\vartheta}},
\vspace{-0.2cm}
\end{equation}
\noindent where $\alpha_\delta = -\ln\!\left(1 - \delta/\max\{V_2(0,0), 2\delta\} \!\right)$. The strict decrease of $V_2$ along flows and across jumps ensures that every solution $\mathbf{x}_{\boldsymbol{\vartheta}}(t,j)$ to $\mathcal{H}_{\boldsymbol{\vartheta}}$ remains confined to $\mathcal{V}_{\mathbf{2}}$ for all $(t,j) \in \operatorname{dom} \mathbf{x}_{\boldsymbol{\vartheta}}$. Furthermore, the jump map $\mathbf{G}_{\boldsymbol{\vartheta}}(\mathbf{D}_{\boldsymbol{\vartheta}}) \subset \mathbf{C}_{\boldsymbol{\vartheta}}$ precludes trajectories from escaping $\mathbf{C}_{\boldsymbol{\vartheta}} \cup \mathbf{D}_{\boldsymbol{\vartheta}}$, establishing that any maximal solution to $\mathcal{H}_{\boldsymbol{\vartheta}}$ is bounded and complete \cite[Proposition 6.10]{goebel_2012}. Then, based on \cite[Theorem 3.18]{goebel_2012}, $\mathcal{A}_{\boldsymbol{\vartheta}}$ is uniformly globally asymptotically stable for $\mathcal{H}_{\boldsymbol{\vartheta}}$. By applying a logic similar to the one followed in \hyperref[thm:StabilityPositionSystem]{Theorem~\ref{thm:StabilityPositionSystem}}, it follows that  $\|\boldsymbol{\Tilde{\omega}}(t,j)\| \leq \alpha_\omega \; \forall \; (t,j) \in \textrm{dom } \mathbf{x}_{\boldsymbol{\vartheta}}$, with $\alpha_\omega \in \mathbb{R}_{\geq0}$ established directly over the compact set $\mathcal{V}_{\mathbf{2}}$ via
\begin{equation*}
    \alpha_\omega \coloneqq \max_{\mathbf{x}_{\boldsymbol{\vartheta}} \in \mathcal{V}_{\mathbf{2}}} \, \|\boldsymbol{\Tilde{\omega}}(t,j)\|.
\end{equation*}
Then, $\|\boldsymbol{s_\omega}(k_\omega \boldsymbol{\Tilde{\omega}})\| \geq \gamma_\omega \|\boldsymbol{\Tilde{\omega}}\|$ holds with $\gamma_\omega \coloneqq k_\omega M_\omega(M_\omega^2 + (k_\omega \alpha_\omega)^2)^{-1/2}$. Recalling \eqref{eq:BoundV2dot} and defining $\mathbf{B}_{\boldsymbol{\vartheta}}^* \coloneqq \min\{\gamma_\omega^2, 1\}\mathbf{B}_{\boldsymbol{\vartheta}}$, for all $\mathbf{x}_{\boldsymbol{\vartheta}} \in \mathbf{C}_{\boldsymbol{\vartheta}}$ one has \begin{equation*}
    \dot{V}_2 \!\le\!  - \alpha_{\dot{V}_2}\|\mathbf{\Tilde{x}_{\boldsymbol{\vartheta}}}\|^2, \quad \textrm{with } \quad \alpha_{\dot{V}_2} = \min_{\mathbf{x}_{\boldsymbol{\vartheta}} \in \mathcal{V}_{\mathbf{2}}} \lambda_{\min}(\mathbf{B}_{\boldsymbol{\vartheta}}^*).
\end{equation*}
Thus, 
\begin{subequations}
\label{eq:expVconditions}
    \begin{align}
        \dot{V}_2(\mathbf{x}_{\boldsymbol{\vartheta}}) &\leq -\lambda_{\vartheta} V_2(\mathbf{x}_{\boldsymbol{\vartheta}}), &&\forall \quad \mathbf{x}_{\boldsymbol{\vartheta}} \in \mathbf{C}_{\boldsymbol{\vartheta}}, \\
        V_2(\mathbf{G}_{\boldsymbol{\vartheta}}(\mathbf{x}_{\boldsymbol{\vartheta}})) &\leq e^{-\lambda_\vartheta} V_2(\mathbf{x}_{\boldsymbol{\vartheta}}), &&\forall \quad \mathbf{x}_{\boldsymbol{\vartheta}} \in \mathbf{D}_{\boldsymbol{\vartheta}},
    \end{align}
\end{subequations}
where $\lambda_{\vartheta} \coloneqq \min\{\alpha_{\dot{V}_2}^{-1}\lambda_{\max}(\mathbf{A_2}), \alpha_\delta\}$. Consequently, invoking \cite[Theorem 1]{teel2012}, for any initial condition $\mathbf{x}_{\boldsymbol{\vartheta}}(0,0)$ originating in an arbitrary compact set $\boldsymbol{\Omega_\vartheta} \subset \boldsymbol{\chi_\vartheta}$, the compact set $\mathcal{A}_{\boldsymbol{\vartheta}}$ is globally asymptotically stable and robustly semi-globally exponentially stable for $\mathcal{H}_{\boldsymbol{\vartheta}}$.
\par Leveraging the bound $\|\boldsymbol{\Tilde{\vartheta}}\| \leq \alpha_\vartheta$, the limits of the radial saturation functions $\boldsymbol{s_\vartheta}$ and $\boldsymbol{s_\omega}$, and the invariant norms $\|\mathbf{\Tilde{R}^{\!\top}}\boldsymbol{\omega}_{\mathbf{d}}\| = \|\boldsymbol{\omega}_{\mathbf{d}}\|$ and $\|\mathbf{\Tilde{R}^{\!\top}}\boldsymbol{\dot{\omega}}_{\mathbf{d}}\| = \|\boldsymbol{\dot{\omega}}_{\mathbf{d}}\|$, an application of the triangle inequality alongside \hyperref[lem:AttitudeReferencesBounds]{Lemma~\ref*{lem:AttitudeReferencesBounds}} yields
\begin{equation*}
    \mathbf{e_i^{\!\top}}\boldsymbol{\tau} > J_{ii} U_{\dot{\omega}_d} + \tfrac{1}{2} \left| J_{jj} - J_{kk} \right| U_{\omega_d}^2 + \iota,
\end{equation*}
for all $i \in \{1,2,3\}$ and mutually distinct indices $\{j, k\} = \{1, 2, 3\} \setminus \{i\}$. Hence, $\boldsymbol{\tau} \in \boldsymbol{\Omega_\tau}$ for any solution $\mathbf{x}_{\boldsymbol{\vartheta}}(t,j)$ to $\mathcal{H}_{\boldsymbol{\vartheta}}$.
\end{proof}

\par Similar to the approach in \cite{martins2023robust}, the sequel leverages the exponential stability established in \hyperref[thm:StabilityPositionSystem]{Theorem~\ref*{thm:StabilityPositionSystem}} and \hyperref[thm:StabilityAttitudeTrackingSystemMRP]{Theorem~\ref*{thm:StabilityAttitudeTrackingSystemMRP}}, alongside external $\mathcal{L}_2$ stability arguments \cite{martins2024integrator}, to guarantee robust global tracking for the interconnected system. Let $\mathbf{x_{h}} \coloneqq (\mathbf{r}, \mathbf{\Tilde{x}}) \in \boldsymbol{\chi}_{\mathbf{1}}$, with $\mathbf{\Tilde{x}} \coloneqq (\mathbf{\Tilde{x}_p}, \mathbf{\Tilde{x}}_{\boldsymbol{\vartheta}})$ and $\boldsymbol{\chi}_{\mathbf{1}} \coloneqq \boldsymbol{\Omega}_{\mathbf{r}} \times \mathbb{R}^{6} \times M_p\mathbb{B}^6_\infty \times \boldsymbol{\chi_\vartheta}$ to formulate the hybrid system $\mathcal{H} = \left(\mathbf{C},\mathbf{F},\mathbf{D},\mathbf{G}\right)$:
\vspace{-0.2cm}
\begin{subequations}
\label{eq:HybridSystemErrorDynamics}
\begin{equation}
    \mathbf{F}(\mathbf{x_{h}}) \!\coloneqq\! \left(\!\!\!\begin{array}{c}
        \mathbf{F_r}(\mathbf{r}), \!\;
        \mathbf{\Tilde{v}}, \!\;
     \mathbf{u_f} + \mathbf{d}, \!\;\mathbf{F_f}(\mathbf{x_f}, \mathbf{\bar{u}_p}), \!\;
        \mathbf{F}_{\boldsymbol{\Tilde{\vartheta}}}(\mathbf{x}_{\boldsymbol{\vartheta}})
        \end{array}\!\!\!\right)
\vspace{-0.2cm}
\end{equation}
\begin{equation}
    \mathbf{C}(\mathbf{x_h}) \coloneqq \left\{ \mathbf{x_{h}} \in \boldsymbol{\chi}_{\mathbf{1}}: \|\boldsymbol{\Tilde{\vartheta}}\| \leq 1 + \delta \right\}
    \vspace{-0.2cm}
\end{equation}
\begin{equation}
    \mathbf{G}(\mathbf{x_h}) \coloneqq (\mathbf{r}, \mathbf{\Tilde{x}_p}, \boldsymbol{\Upsilon}(\boldsymbol{\Tilde{\vartheta}}), \boldsymbol{\Tilde{\omega}})
    \vspace{-0.2cm}
\end{equation}
\begin{equation}
     \mathbf{D}(\mathbf{x_h}) \coloneqq \left\{ \mathbf{x_h} \in \boldsymbol{\chi}_{\mathbf{1}}: \|\boldsymbol{\Tilde{\vartheta}}\| = 1 + \delta \right\}
     \vspace{-0.2cm}
\end{equation}
\end{subequations}
\noindent The perturbation $\mathbf{d} \coloneqq \mathbf{R_d}(\mathbf{\Tilde{R}} - \mathbf{I_3})\|\mathbf{u_p}\|\mathbf{e_3}$ dictates the cascaded coupling between the rotational and translational dynamics. Dominating this term is strictly necessary to establish the stability of the interconnected system.

\begin{thm}
\label{thm:StabilityFullTrackingSystem}
Suppose the conditions of \hyperref[assump:ReferenceTrajectory]{Assumption~\ref*{assump:ReferenceTrajectory}} hold for all $t \geq 0$. Then, the compact set $\mathcal{A}_{\mathbf{1}} \coloneqq \{\mathbf{x_h} \in \boldsymbol{\chi_1}: \mathbf{\Tilde{x}_p} \in \mathcal{A}_{\mathbf{p}}, \mathbf{\Tilde{x}}_{\boldsymbol{\vartheta}} \in \mathcal{A}_{\boldsymbol{\vartheta}}\}$ is globally asymptotically stable and semi-globally exponentially stable for $\mathcal{H}$.
\end{thm}
\begin{proof} 
Let $\mathbf{x_z} \coloneqq (\mathbf{r}, \mathbf{\Tilde{x}_z}) \in \boldsymbol{\chi}_{\mathbf{1}}$, with $\mathbf{\Tilde{x}_z} \coloneqq (\mathbf{z}, \mathbf{\Tilde{x}}_{\boldsymbol{\vartheta}})$ and $\boldsymbol{\chi}_{\mathbf{1}} \coloneqq \boldsymbol{\Omega}_{\mathbf{r}} \times \mathbb{R}^{6} \times M_p\mathbb{B}^6_\infty \times \boldsymbol{\chi_\vartheta}$ to formulate the hybrid system $\mathcal{H}_{\mathbf{z}} = \left(\mathbf{C}_{\mathbf{z}},\mathbf{F}_{\mathbf{z}},\mathbf{D}_{\mathbf{z}},\mathbf{G}_{\mathbf{z}}\right)$:
\vspace{-0.2cm}
\begin{subequations}
\label{eq:HybridSystemErrorDynamicsZ}
\begin{equation}
    \mathbf{F}_{\mathbf{z}}(\mathbf{x_z}) \!\coloneqq\! \left(\!\!\!\begin{array}{c}
        \mathbf{F_r}(\mathbf{r}), \!\;
        \mathbf{z_2}, \!\;
     \mathbf{\bar{u}_p} + \mathbf{d}, \!\;\mathbf{F_f}(\mathbf{x_f}, \mathbf{\bar{u}_p}), \!\;
        \mathbf{F}_{\boldsymbol{\Tilde{\vartheta}}}(\mathbf{x}_{\boldsymbol{\vartheta}})
        \end{array}\!\!\!\right)
\vspace{-0.2cm}
\end{equation}
\begin{equation}
    \mathbf{C}_{\mathbf{z}}(\mathbf{x_z}) \coloneqq \left\{ \mathbf{x_z} \in \boldsymbol{\chi}_{\mathbf{1}}: \|\boldsymbol{\Tilde{\vartheta}}\| \leq 1 + \delta \right\}
\end{equation}
\begin{equation}
    \mathbf{G}_{\mathbf{z}}(\mathbf{x_z}) \coloneqq (\mathbf{r}, \mathbf{\Tilde{x}_p}, \boldsymbol{\Upsilon}(\boldsymbol{\Tilde{\vartheta}}), \boldsymbol{\Tilde{\omega}})
    \vspace{-0.2cm}
\end{equation}
\begin{equation}
     \mathbf{D}_{\mathbf{z}}(\mathbf{x_z}) \coloneqq \left\{ \mathbf{x_z} \in \boldsymbol{\chi}_{\mathbf{1}}: \|\boldsymbol{\Tilde{\vartheta}}\| = 1 + \delta \right\}
\end{equation}
\end{subequations}
The hybrid system $\mathcal{H}_{\mathbf{z}}$ formulates the cascaded interconnection between the attitude and position tracking loops, with the translational dynamics expressed in $\mathbf{z}$-coordinates. Furthermore, $\mathcal{H}_{\mathbf{z}}$ satisfies the hybrid basic conditions \cite[Assumption 6.5]{goebel_2012}, strictly guaranteeing its well-posedness \cite[Theorem 6.30]{goebel_2012}. The exponential stability established in \hyperref[thm:StabilityAttitudeTrackingSystemMRP]{Theorem~\ref*{thm:StabilityAttitudeTrackingSystemMRP}} guarantees that for all $(t,j) \in \operatorname{dom} \mathbf{x_z}$,
\begin{equation}
    \label{eq:MRPBoundExponential}
    \|\boldsymbol{\Tilde{\vartheta}}(t,j)\| \leq \left(\tfrac{\lambda_{\max}(\mathbf{A}_{\boldsymbol{\vartheta_1}})}{\lambda_{\min}(\mathbf{A}_{\boldsymbol{\vartheta_2}})} \right)^{1/2} \!\!\! {e}^{-\lambda_{\vartheta}(t+j)/2 } \|\mathbf{\Tilde{x}}_{\boldsymbol{\vartheta}}(0,0)\|
\end{equation}
Bounding the perturbation vector via the Frobenius norm yields $\|\mathbf{d}\| \leq \sqrt{2\operatorname{Tr}(\mathbf{I_3} - \mathbf{\Tilde{R}})} \|\mathbf{u}\|$. Invoking \eqref{eq:rotationmatrixMRP}, this bound simplifies to
\vspace{-0.2cm}
\begin{equation}
\label{eq:InterconnectionBound}
    \|\mathbf{d}\| \leq 4\sqrt{2}\|\boldsymbol{\Tilde{\vartheta}}_{\downarrow t}\| \|\mathbf{u}\|
\vspace{-0.2cm}
\end{equation}
Consequently, the exponential decay of $\|\boldsymbol{\Tilde{\vartheta}}\|$ establishes that $\mathbf{d}$ is square-integrable, establishing the finite $\mathcal{L}_2$ norm
\vspace{-0.2cm}
\begin{equation*}
    \left(\int_0^\infty \|\mathbf{d}(t)\|^2 dt \right)^{1/2} < \infty
    \vspace{-0.2cm}
\end{equation*} 
Furthermore, \eqref{eq:MRPBoundExponential} guarantees that the attitude error enters any arbitrarily small neighborhood of the origin in finite time. In this way, it follows from \cite[Theorem 4]{martins2024integrator} that the trajectories of
\begin{equation*}
    \left[\begin{array}{c}
         \mathbf{\dot{z}_1}   \\
         \mathbf{\dot{z}_2} 
    \end{array}\right] = \mathbf{F_z}(\mathbf{z_1}, \mathbf{z_2}, \mathbf{\bar{u}_p}) = \left[ \begin{array}{c}
         \mathbf{z_2}  \\
         \mathbf{\bar{u}_p} + \mathbf{d}
    \end{array}\right].
\end{equation*}
are bounded. Thus, there exists $\alpha_1 \in \mathbb{R}_{\geq0}$ such that $\|k_p\mathbf{z_1}(t,j) \!+\! k_v\mathbf{z_2}(t,j)\| \leq \alpha_1 \!\; \forall \!\; (t,j) \!\in\! \textrm{dom } \mathbf{x_z}$. 
\par Let $V_3\colon \boldsymbol{\chi_1} \to \mathbb{R}_{\geq 0}$ be given by $V_3(\mathbf{x_z}) = k_1V_1(\mathbf{z}) + V_2(\mathbf{\Tilde{x}}_{\boldsymbol{\vartheta}}),$ with $k_1 \in \mathbb{R}_{>0}$ satisfying
\begin{equation}
    \label{eq:k1bound}
    k_1 < \min\left\{\frac{1}{2k_v},\frac{\sigma'(\alpha_1)}{k_p}\right\} \frac{b k_\omega \gamma_\vartheta \sigma'_\omega(k_\omega \alpha_\omega)}{32^2} \frac{m^2}{T_{\min}^2}.
\end{equation}
From the proofs of \hyperref[thm:StabilityPositionSystem]{Theorem~\ref*{thm:StabilityPositionSystem}} and \hyperref[thm:StabilityAttitudeTrackingSystemMRP]{Theorem~\ref*{thm:StabilityAttitudeTrackingSystemMRP}}, it follows that $V_3$ is continuously differentiable on $\boldsymbol{\chi_1}$, radially unbounded, and positive-definite with respect to $\mathcal{A}_{\mathbf{z}} \coloneqq \{\mathbf{x_z} \in \boldsymbol{\chi_1}\colon \mathbf{\Tilde{x_z}} = \boldsymbol{0} \}$. Considering \eqref{eq:V1dotbound} and \eqref{eq:BoundV2dot}, the time derivative of $V_3$ yields
\vspace{-0.2cm}
\begin{equation*}
\begin{aligned}
    \dot{V}_3 \leq & -\tfrac{k_1}{2}\mathbf{\Tilde{x}_p^{\!\top}}\mathbf{B_p}(\alpha_1)\mathbf{\Tilde{x}_p} - \mathbf{\Tilde{x}}_{\boldsymbol{\vartheta}}^{\!\top}\mathbf{B}_{\boldsymbol{\vartheta}}^*\mathbf{\Tilde{x}}_{\boldsymbol{\vartheta}} - \frac{b}{8}k_\omega\gamma_\vartheta\boldsymbol{\Tilde{\vartheta}}^{\!\top}\boldsymbol{\Lambda}(\boldsymbol{\Tilde{\omega}}) \boldsymbol{\Tilde{\vartheta}} \\ 
    & - k_1\tfrac{k_p \sigma'(\alpha_1)}{2} \|\mathbf{z_2}\|^2 - k_1\tfrac{k_v}{4}\|\mathbf{\bar{u}_p}\|^2 \\ & 
    + 2k_1k_p\|\mathbf{z_2}\|\|\mathbf{d}\| + k_1k_v\|\mathbf{\bar{u}_p}\|\|\mathbf{d}\|.
\end{aligned}
\end{equation*}
In virtue of \eqref{eq:InterconnectionBound} and \eqref{eq:k1bound}, the following inequality holds
\begin{equation*}
        \dot{V}_3 \leq -\tfrac{k_1}{2}\mathbf{\Tilde{x}_p^{\!\top}}\mathbf{B_p}(\alpha_1)\mathbf{\Tilde{x}_p} -\mathbf{\Tilde{x}}_{\boldsymbol{\vartheta}}^{\!\top}\mathbf{B}_{\boldsymbol{\vartheta}}^*\mathbf{\Tilde{x}}_{\boldsymbol{\vartheta}}, 
\end{equation*}
which leads to

\begin{equation*}
    \dot{V}_3 \!\leq\! -\min\!\left\{\tfrac{k_1}{2}\lambda_{\min}(\mathbf{B_p}(\alpha_1)), \lambda_{\min}(\mathbf{B}_{\boldsymbol{\vartheta}}^*)\right\} \|\mathbf{\Tilde{x}}\|^2.
\end{equation*}
Furthermore, since
\begin{equation}
    V_3 \leq \max \left\{k_1\lambda_{\max}\left(\mathbf{A_p}\right), \lambda_{\max}\left(\mathbf{A}_{\boldsymbol{\vartheta_1}}\right)\right\} \|\mathbf{\Tilde{z}}\|^2
\end{equation}
it follows that, for all $\mathbf{x_z} \in \mathbf{C_z}$,
\begin{equation}
    \label{eq:V3FlowBehavior}
    \dot{V}_3 \!\leq\! \! - \frac{\min\!\left\{\!k_1\!2^{-1}\!\lambda_{\min}\!(\mathbf{B_p}(\alpha_1\!)\!), \lambda_{\min}\!(\mathbf{B}_{\boldsymbol{\vartheta}}^*\!)\!\right\}}{\max\left\{k_1\lambda_{\max}\left(\mathbf{A_p}\right), \lambda_{\max}\left(\mathbf{A}_{\boldsymbol{\vartheta_1}}\right)\right\}} V_3 \!=\! - \lambda_{\!f} V_3
\end{equation}
During jumps, $V_3$ yields
\begin{equation}
    \label{eq:V3JumpBehavior}
    V_3(\mathbf{G_z}(\mathbf{x_z})) \leq e^{-\lambda_{\!j}} V_3(\mathbf{x_z}) \quad \forall \quad \mathbf{x_z} \in \mathbf{D_z}
\end{equation}
with $\lambda_{\!j} = -\ln\!\left(1 - \delta/\max\{V_3(0,0), 2\delta\} \!\right)$. Thus, $V_3$ strictly decreases during both flows and jumps, which implies that any solution $\mathbf{x_z}\!\left(t,\;j\right)$ to $\mathcal{H}_{\mathbf{z}}$ remains in $\mathcal{V}_{\mathbf{3}} = \{\mathbf{x_z} \in \boldsymbol{\chi_1}: V_3(\mathbf{x_z}) \leq V_3(\mathbf{x_z}(0,0))\}$ for all $(t,j) \in \textrm{dom}\; \mathbf{x_z}$. In this way, let $k_1^* \in \mathbb{R}_{>0}$ be given by
\begin{equation*}
        k_1^* = \min\left\{\frac{1}{2k_v},\frac{1}{ k_p}\right\} \frac{b k_\omega k_\vartheta}{32^2} \frac{m^2}{T^2_{\min}},
\end{equation*}
which yields $k_1^*>k_1$, and consider the function $V_3^* = k_1^*V_1 + V_2$ . Then, the bound $\alpha_1$ is computed analogously to $\alpha_p$:
\begin{equation*}
    \alpha_1 \coloneqq \max_{\mathbf{x_z} \in \mathcal{V}_{\mathbf{3}}} \, \max \big\{ \|k_p\mathbf{z_1} + k_v \mathbf{z_2}\|, \|k_p\mathbf{z_1}\| \big\}.
\end{equation*}
Furthermore, the jump inclusion $\mathbf{G_z}(\mathbf{D_z}) \subset \mathbf{C_z}$ precludes trajectories from escaping $\mathbf{C_z} \cup \mathbf{D_z}$, guaranteeing that all maximal solutions to $\mathcal{H}_{\mathbf{z}}$ are bounded and complete \cite[Proposition 6.10]{goebel_2012}. Consequently, \eqref{eq:V3FlowBehavior} and \eqref{eq:V3JumpBehavior} lead to
\begin{equation*}
    \dot{V_3} \leq -\lambda V_3 \quad \forall \quad \mathbf{x_z} \in \mathbf{C_z}
\vspace{-0.2cm}
\end{equation*}
\begin{equation*}
    V_3(\mathbf{G}(\mathbf{x_z})) \leq e^{-\lambda} V_3(\mathbf{x_z}) \quad \forall \quad \mathbf{x_z} \in \mathbf{D_z}
\vspace{-0.2cm}
\end{equation*}
with $\lambda = \min\{\lambda_{\!f}, \lambda_{\!j}\}$. Therefore, invoking \cite[Theorem 3.18]{goebel_2012} and \cite[Theorem 1]{teel2012} and considering the linear transformation \eqref{eq:Zcoordinates} establishes the set $\mathcal{A}_{\mathbf{1}}$ as globally asymptotically and semi-globally exponentially stable for $\mathcal{H}$.
\end{proof}

\par The path-lifting mechanism and stability equivalence from \cite{martins2025hybrid} allow the MRP controller to operate directly on $\mathrm{SO}(3)$ while retaining its formal tracking guarantees. The complete closed-loop hybrid system $\mathcal{H}_{\mathbf{R}} \coloneqq (\mathbf{C_R}, \mathbf{F_R}, \mathbf{D_R}, \mathbf{G_R})$ integrates the dynamics through the state vector $\mathbf{x_R} \coloneqq (\mathbf{r}, \mathbf{\Tilde{x}_p}, \mathbf{x_l}, \boldsymbol{\Tilde{\omega}}) \in \boldsymbol{\chi}_{\mathbf{R}}$, with state space $\boldsymbol{\chi}_{\mathbf{R}} \coloneqq \boldsymbol{\Omega}_{\mathbf{r}} \times \mathbb{R}^6 \times M_p\mathbb{B}^{6}_\infty \times \boldsymbol{\chi}_{\mathbf{l}} \times \mathbb{R}^3$ and data
\vspace{-0.1500cm}
\begin{subequations}
\begin{equation*}
    \mathbf{C_R} \!\coloneqq\! \left\{\mathbf{x_R} \!\in\! \boldsymbol{\chi}_{\mathbf{R}}: \mathbf{x_l} \in \mathbf{C_l} \right\}
\vspace{-0.10cm}
\end{equation*}
\begin{equation*}
   \mathbf{F_R} (\mathbf{x_R}) \!\coloneqq\! \left(\!\!\begin{array}{c}
        \mathbf{F_r}(\mathbf{r}) \\
        \mathbf{\Tilde{v}} \\
        \mathbf{u_f} + \mathbf{d} \\
        \mathbf{F_f}(\mathbf{x_f}, \mathbf{\bar{u}_p}) \\
        \boldsymbol{0} \\ 
        0 \\        \mathbf{\Tilde{R}}\left[\boldsymbol{\Tilde{\omega}}\right]_\times \\
        \mathbf{J}^{-1}\left(\boldsymbol{\Delta}\left(\boldsymbol{\Tilde{\omega}}, \mathbf{\Tilde{R}^\top}\boldsymbol{\omega}_\mathbf{d} \right) + \boldsymbol{\tau}_{\!\mathbf{R}}(\mathbf{x_R})) + \boldsymbol{\tau_c}\right) \\
        \end{array}\!\!\right)
\vspace{-0.10cm}
\end{equation*}
\begin{equation*}
     \mathbf{D_R} \!\coloneqq\! \left\{\mathbf{x_R} \in \boldsymbol{\chi}_{\mathbf{R}}\!\colon\! \mathbf{x_l} \in \mathbf{D_l} \right\}, \quad \!\!\! \mathbf{G_R} (\mathbf{x_R}) \!\coloneqq\! (\mathbf{r}, \mathbf{\Tilde{x}_p}, \mathbf{G_l}(\mathbf{x_l}), \boldsymbol{\Tilde{\omega}}),
\vspace{-0.10cm}
\end{equation*}
\end{subequations}
with $\boldsymbol{\tau}_{\!\mathbf{R}}(\mathbf{x_R}): \boldsymbol{\chi}_{\mathbf{R}} \to \mathbb{R}^3$ satisfying
\vspace{-0.1500cm}
\begin{equation*}
    \boldsymbol{\tau}_{\!\mathbf{R}}(\mathbf{x_R}) \coloneqq \boldsymbol{\tau}(\boldsymbol{\varphi}(m^*\boldsymbol{\Phi}(\mathbf{\hat{q}}, \mathbf{\Tilde{R}})), \boldsymbol{\Tilde{\omega}}).
    \vspace{-0.1500cm}
\end{equation*}
\begin{thm}
\label{thm:StabilityTrackingSystemR}
The hybrid system $\mathcal{H}_{\mathbf{R}}$ is well-posed. Suppose that \hyperref[assump:ReferenceTrajectory]{Assumption~\ref*{assump:ReferenceTrajectory}} holds. Then, the set $\mathcal{A}_{\mathbf{R}} \coloneqq \{\mathbf{x_R} \in \boldsymbol{\chi}_{\mathbf{R}}: \mathbf{\Tilde{x}_p} = \boldsymbol{0}, \mathbf{\Tilde{R}} = \mathbf{I_3},\,\boldsymbol{\Tilde{\omega}} = \boldsymbol{0}\}$ is robustly globally asymptotically stable and semi-globally exponentially stable for $\mathcal{H}_{\mathbf{R}}$. Furthermore, $T \in \boldsymbol{\Omega}_{T}$ and $\boldsymbol{\tau}_{\mathbf{R}} \in \boldsymbol{\Omega}_{\mathbf{R}}$ for all $(t,j) \in \textrm{dom } \mathbf{x_R}$.
\end{thm}
\begin{proof} The autonomous hybrid system $\mathcal{H}_{\mathbf{l}}$ is well-posed \cite{martins2025hybrid}. The sets $\mathbf{C_R}$ and $\mathbf{D_R}$ are closed subsets of $\boldsymbol{\chi}_\mathbf{R}$. Formulated from continuous differential equations and a bounded, convex set-valued mapping, the flow map $\mathbf{F_R}$ is inherently outer semicontinuous, locally bounded relative to $\mathbf{C_R} \subset \operatorname{dom} \mathbf{F_R}$, and convex-valued for all $\mathbf{x_R} \in \mathbf{C_R}$. As the states $\mathbf{r}, \mathbf{\Tilde{x}_p}, \mathbf{\Tilde{R}}$, and $\boldsymbol{\Tilde{\omega}}$ are invariant across jumps and $\mathcal{H}_{\mathbf{l}}$ is well-posed \cite{martins2025hybrid}, the jump map $\mathbf{G_R}$ is outer semicontinuous and locally bounded relative to $\mathbf{D_R} \subset \operatorname{dom} \mathbf{G_R}$. Therefore, $\mathcal{H}_{\mathbf{R}}$ fulfills the hybrid basic conditions \cite[Assumption 6.5]{goebel_2012}, guaranteeing it is well-posed \cite[Theorem 6.30]{goebel_2012}.
\par Based on \hyperref[thm:StabilityFullTrackingSystem]{Theorem~\ref*{thm:StabilityFullTrackingSystem}}, $\mathcal{A}_{\mathbf{1}}$ is asymptotically stable for $\mathcal{H}$ from any $\mathbf{x_h}(0,0) \in \mathcal{B}_{\mathbf{1}} \coloneqq \boldsymbol{\chi}_{\mathbf{1}}$. Through \cite[Theorem 1]{martins2025hybrid} and \eqref{eq:rotationmatrixMRP}, this guarantees $\mathcal{A}_{\mathbf{R}}$ is asymptotically stable for $\mathcal{H}_{\mathbf{R}}$, possessing the basin of attraction $\{\mathbf{x_R} \in \boldsymbol{\chi}_{\mathbf{R}}: \operatorname{dist}(\mathbf{\hat{q}}, \mathcal{Q}(\mathbf{\Tilde{R}})) < 1\}$. Moreover, for any arbitrarily large compact set $\boldsymbol{\Omega}_{\mathbf{1}}$, \hyperref[thm:StabilityFullTrackingSystem]{Theorem~\ref*{thm:StabilityFullTrackingSystem}} ensures exponential stability of $\mathcal{A}_{\mathbf{1}}$ for $\mathbf{x_h}(0,0) \in \mathcal{B}_{\mathbf{1}} = \boldsymbol{\Omega}_{\mathbf{1}} \subset \boldsymbol{\chi}_{\mathbf{1}}$. Applying \cite[Theorem 2]{martins2025hybrid} extends this exponential stability to $\mathcal{A}_{\mathbf{R}}$ for all initial conditions $\mathbf{x_R}(0, 0) \in \mathcal{B} \coloneqq \{\mathbf{x_R} \in \boldsymbol{\chi}_{\mathbf{R}}: (\mathbf{r}, \mathbf{\Tilde{x}_p}, \boldsymbol{\varphi}(m^*\boldsymbol{\Phi}(\mathbf{\hat{q}}, \mathbf{\Tilde{R}})), \boldsymbol{\tilde{\omega}}) \in \boldsymbol{\Omega}_{\mathbf{1}}, \operatorname{dist}(\mathbf{\hat{q}}, \mathcal{Q}(\mathbf{\Tilde{R}})) < 1\}$. 
Furthermore, since $\mathcal{H}_{\mathbf{R}}$ is well-posed, these stability results possess nominal robustness against perturbations, including external disturbances, parameter variations, or measurement noise, characterized by $\mathcal{KL}$ bounds \cite[Definition 7.18]{goebel_2012}. Therefore, $\mathcal{A}_{\mathbf{R}}$ is robustly globally asymptotically and semi-globally exponentially stable for $\mathcal{H}_{\mathbf{R}}$. Finally, \hyperref[thm:StabilityPositionSystem]{Theorems~\ref*{thm:StabilityPositionSystem}} and \hyperref[thm:StabilityAttitudeTrackingSystemMRP]{\ref*{thm:StabilityAttitudeTrackingSystemMRP}} guarantee that the control inputs remain bounded such that $T \in \boldsymbol{\Omega}_T$ and $\boldsymbol{\tau} \in \boldsymbol{\Omega}_{\boldsymbol{\tau}}$ for all $(t,j) \in \operatorname{dom} \mathbf{x_R}$
\end{proof}

\begin{rem}
Based on the definition of the tracking errors, the condition $\mathbf{x_R} \in \mathcal{A}_{\mathbf{R}}$ directly implies $(\mathbf{r}, \mathbf{x}) \in \mathcal{A}$. Consequently, the stability guarantees of \hyperref[thm:StabilityTrackingSystemR]{Theorem~\ref*{thm:StabilityTrackingSystemR}}, achieved via the strictly bounded control inputs $T$ and $\boldsymbol{\tau}$, formally solve \hyperref[prob:Problem]{Problem~\ref*{prob:Problem}}. \hfill $\square$
\end{rem}

\section{Simulation Results}
\label{section6}

\par To evaluate the proposed saturated hybrid controller, numerical simulations of \eqref{eq:FullDynamics} were conducted. The model operates with a $0.01\ \mathrm{s}$ sampling time, assumes a rigid-body mass of $m = 0.460\ \mathrm{kg}$ with an inertia $\mathbf{J} = \operatorname{diag}(2.24, 2.9, 5.3) \times 10^{-3}\ \mathrm{kg \cdot m^2}$, and imposes the actuation bounds $T_{\max} = 7\,\mathrm{N}$ and $\boldsymbol{\tau_{\max}} = (0.5,0.5,0.5) \; \mathrm{Nm}$. To rigorously validate position and heading tracking under strict actuation limits, the vehicle is initialized in a challenging downward-facing attitude and commanded to track the trajectory
\begin{equation*}
\mathbf{p_d} \!=\! \left( \cos(ft), \sin(ft), \cos(ft) \!\right) + \bar{\mathbf{p}}, \!\quad \boldsymbol{\nu}_\mathbf{d} \!=\! \left( \cos(\psi_d), \sin(\psi_d) \!\right)
\end{equation*}
where $f = 2\pi(15)^{-1}\ \mathrm{Hz}$, $\bar{\mathbf{p}} = (-3, 1, 4.5)\ \mathrm{m}$, and $\psi_d(t) = \pi \sin(ft)\ \mathrm{rad}$, which satisfies \hyperref[assump:ReferenceTrajectory]{Assumption~\ref*{assump:ReferenceTrajectory}}. The system initializes at $\mathbf{p}(0) = (5, 5, 10)\ \mathrm{m}$ with an attitude defined by Euler angles $(\varphi, \theta, \psi)(0) = (-179, 0, 100)\ {}^\circ$ and filter states $\mathbf{x_f}(0) = \mathbf{0}\ \mathrm{N}$. Using the smooth saturation $\sigma(s) = M\tanh(s/M)$ for $\bar{\mathbf{u}}_p$ and the parameters from \hyperref[tab:ControlParameters]{Table~\ref*{tab:ControlParameters}}, this configuration limits the control efforts to $T \geq 3.51\ \mathrm{N}$, $T \leq 6.19\ \mathrm{N}$, $|\mathbf{e}_1^\top \boldsymbol{\tau}| \leq 0.29\ \mathrm{N \cdot m}$, $|\mathbf{e}_2^\top \boldsymbol{\tau}| \leq 0.33\ \mathrm{N \cdot m}$, and $|\mathbf{e}_3^\top \boldsymbol{\tau}| \leq 0.48\ \mathrm{N \cdot m}$, as established in \hyperref[thm:StabilityPositionSystem]{Theorems~\ref*{thm:StabilityPositionSystem}} and \hyperref[thm:StabilityAttitudeTrackingSystemMRP]{\ref*{thm:StabilityAttitudeTrackingSystemMRP}}.


\begin{table}[ht]
\centering
\caption{Control gains and parameters.}
\label{tab:ControlParameters}
\setlength{\tabcolsep}{5pt} 
\begin{tabular}{ccccccccccc}
\hline
$k_p$ & $k_v$ & $k_f$ & $k_s$ & $k_\vartheta$ & $k_\omega$ & $M_p$ & $M_\vartheta$ & $M_\omega$ & $\delta$ & $\alpha$ \\ \hline
9 & 6 & 2 & 20 & 100 & 0.1 & 2 & 0.206 & 0.045 & 0.02 & 0.25 \\ \hline
\end{tabular}
\end{table}

\par \hyperref[fig:Sim1]{Figure~\ref*{fig:Sim1}} illustrates the successful tracking response despite significant initial errors. The downward-facing attitude naturally induces a transient $z$-axis position error peak, as the underactuated vehicle must reorient the thrust vector before arresting the descent. However, following this recovery phase, the position and MRP error norms converge to steady-state values below 0.001 m and 0.0001, respectively (\hyperref[fig:Sim1_c]{Figs.~\ref*{fig:Sim1_c}} and \hyperref[fig:Sim1_d]{\ref*{fig:Sim1_d}}), highlighting the efficacy of the proposed scheme.

\begin{figure}[h]
\centering
\begin{subfigure}[t]{0.49\columnwidth}
    \centering
    \includegraphics[width=\textwidth]{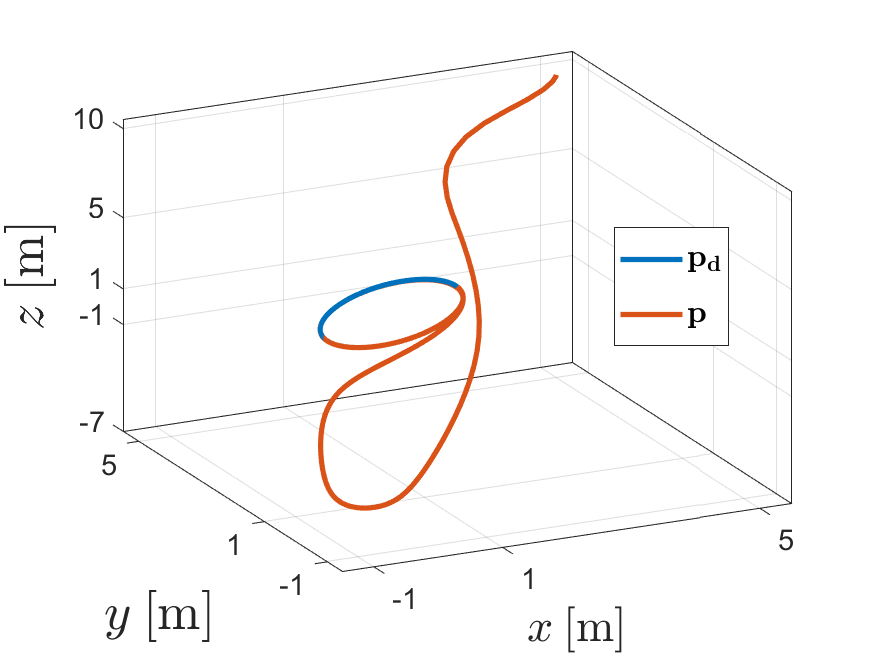}
    \vspace{-0.6cm}
    \caption{Position.}
    \label{fig:Sim1_a}
\end{subfigure}
\begin{subfigure}[t]{0.49\columnwidth}
    \centering
    \includegraphics[width=\textwidth]{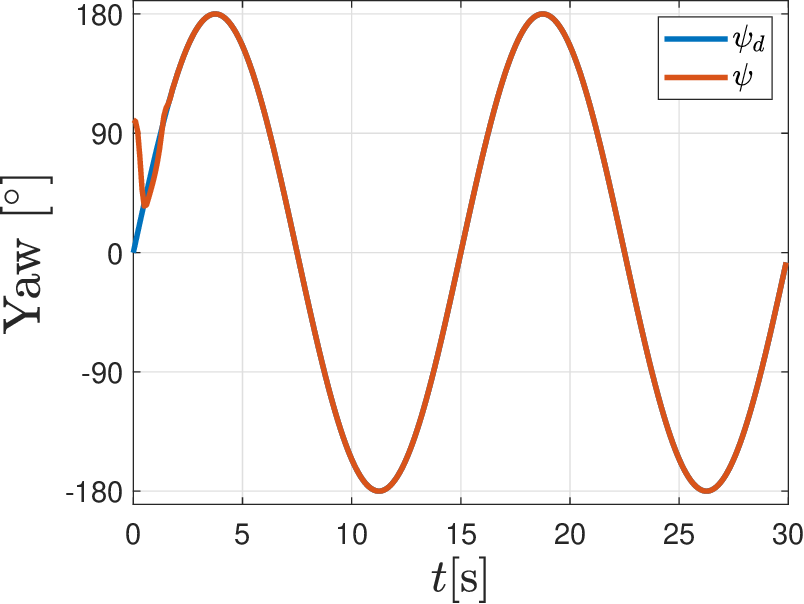}
    \vspace{-0.6cm}
    \caption{Yaw.}
    \label{fig:Sim1_b}
\end{subfigure}
\begin{subfigure}[t]{0.49\columnwidth}
    \centering
    \includegraphics[width=\textwidth]{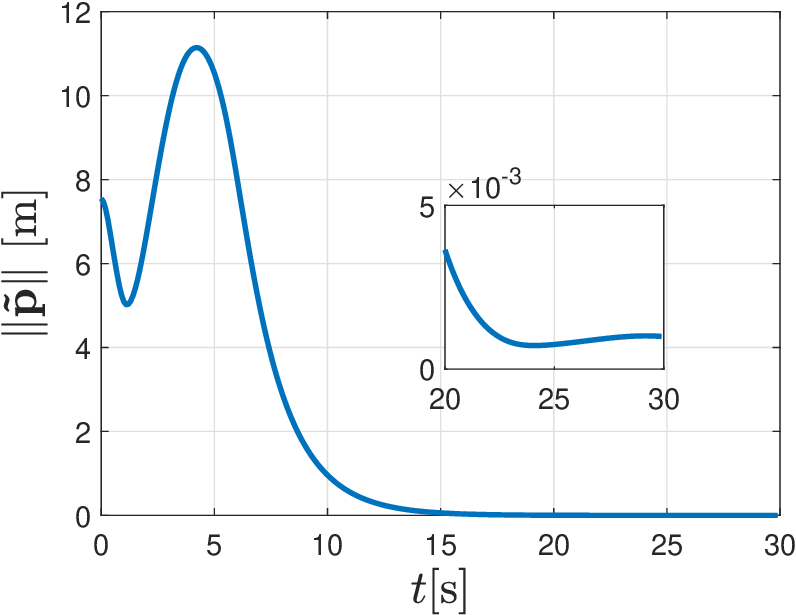}
    \vspace{-0.6cm}
    \caption{Position error norm.}
    \label{fig:Sim1_c}
\end{subfigure}
\hspace{-0.2cm}
\begin{subfigure}[t]{0.49\columnwidth}
    \centering
    \includegraphics[width=\textwidth]{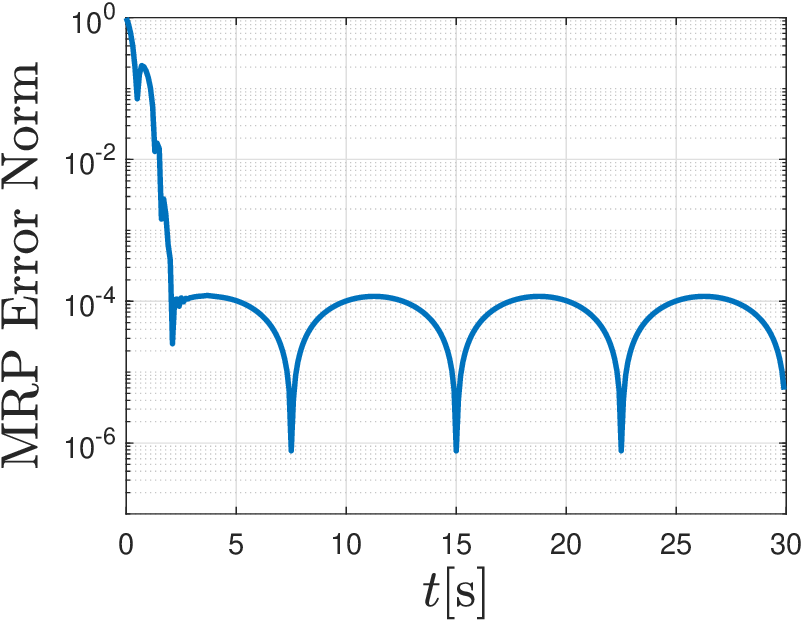}
    \vspace{-0.6cm}
    \caption{MRP error norm.}
    \label{fig:Sim1_d}
\end{subfigure}
\hspace{-0.5cm}
\vspace{-0.1500cm}
\caption{Position and attitude responses in simulation.} \label{fig:Sim1}
\vspace{-0.5cm}
\end{figure}
\hyperref[fig:Sim2]{Figure~\ref*{fig:Sim2}} depicts the actuation alongside the saturation functions signals and filter states. The generated thrust and moments satisfy the theoretical limits, validating the bounds established in \hyperref[thm:StabilityPositionSystem]{Theorems~\ref*{thm:StabilityPositionSystem}} and \hyperref[thm:StabilityAttitudeTrackingSystemMRP]{\ref*{thm:StabilityAttitudeTrackingSystemMRP}}. Driven by the substantial initial errors, $\mathbf{\bar{u}_p}$, $\boldsymbol{s_\vartheta}(k_\vartheta\boldsymbol{\Tilde{\vartheta}})$, and $\boldsymbol{s_\omega}(k_\omega\boldsymbol{\Tilde{\omega}})$ initially operate at thee respective saturation limits. Furthermore, given the initial condition $\mathbf{x_f}= \boldsymbol{0}$ and the linear filter dynamics, \hyperref[fig:Sim2_c]{Fig.~\ref*{fig:Sim2_c}} confirms the filter states respect the same bounds as $\mathbf{\bar{u}_p}$. A small filter gain $k_f$ relative to $k_s$ induces a noticeable transient lag in $\mathbf{u_f}$. Regardless, the strategy achieves precise trajectory tracking despite the delayed, saturated operation, underscoring the robustness of the approach.

\begin{figure}[h]
\centering
\begin{subfigure}[t]{0.49\columnwidth}
    \centering
    \includegraphics[width=\textwidth]{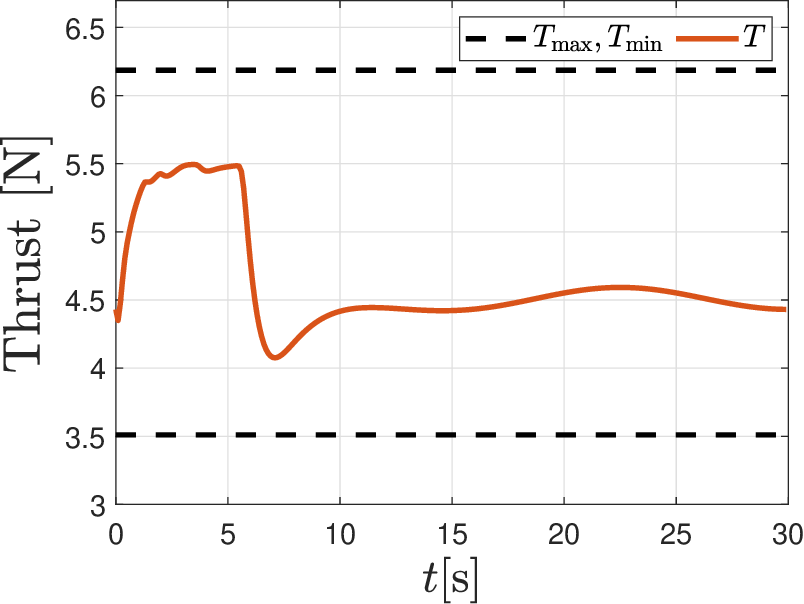}
    \vspace{-0.6cm}
    \caption{Thrust.}
    \label{fig:Sim2_a}
\end{subfigure}
\begin{subfigure}[t]{0.49\columnwidth}
    \centering
    \includegraphics[width=\textwidth]{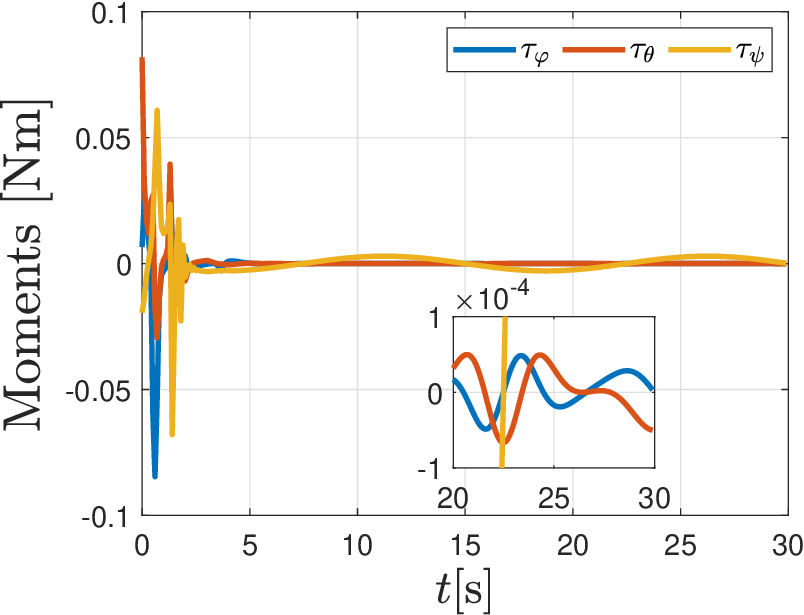}
    \vspace{-0.6cm}
    \caption{Moments.}
    \label{fig:Sim2_b}
\end{subfigure}
\begin{subfigure}[t]{0.49\columnwidth}
    \centering
    \includegraphics[width=\textwidth]{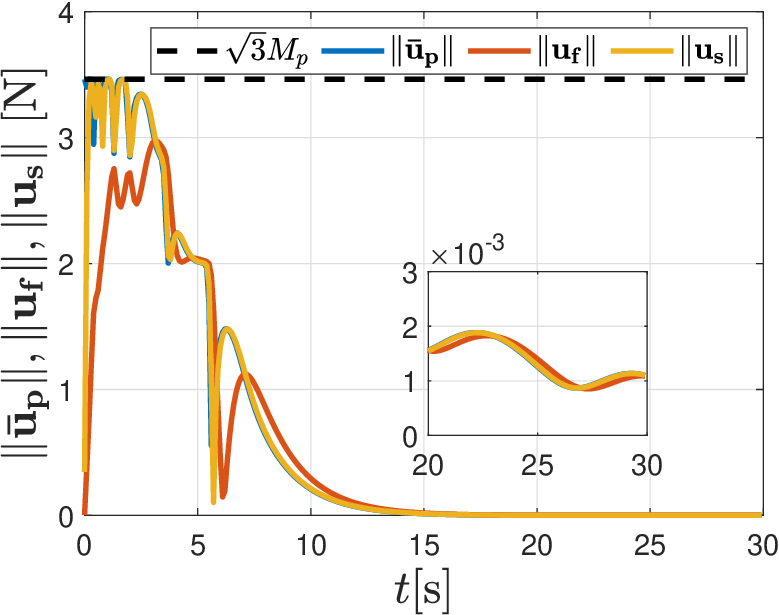}
    \vspace{-0.6cm}
    \caption{Norm of $\mathbf{\bar{u}_p}$, $\mathbf{u_f}$ and $\mathbf{u_s}$.}
    \label{fig:Sim2_c}
\end{subfigure}
\hspace{-0.2cm}
\begin{subfigure}[t]{0.49\columnwidth}
    \centering
    \includegraphics[width=\textwidth]{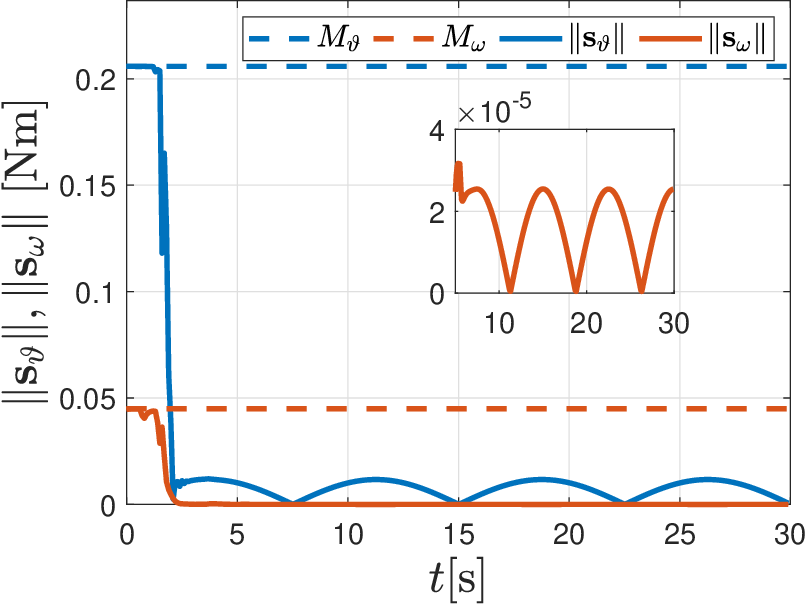}
    \vspace{-0.6cm}
    \caption{Norm of $\boldsymbol{s_\vartheta}$ and $\boldsymbol{s_\omega}$.}
    \label{fig:Sim2_d}
\end{subfigure}
\hspace{-0.5cm}
\vspace{-0.1500cm}
\caption{Actuation and saturation functions in simulation.} \label{fig:Sim2}
\vspace{-0.5cm}
\end{figure}

\section{Conclusion}
\label{section7}

\par This work proposed a robust global position and heading tracking strategy for underactuated vehicles operating under strict, user-defined actuation limits. The hierarchical architecture augments a saturated position controller with first-order filters to generate bounded attitude references. By leveraging a forward-propagating coordinate transformation, global asymptotic and semi-global exponential tracking for the outer loop are formally proven. A saturated, MRP-based inner loop paired with a hybrid dynamic path-lifting mechanism tracks the generated references on a covering space of the attitude configuration manifold. Exploiting a stability equivalence framework, the resulting interconnected strategy guarantees robust global asymptotic and semi-global exponential tracking on $\mathrm{SE}(3)$. Numerical simulations validate the approach, demonstrating precise trajectory tracking under saturation and filter transient lag, successful recovery from extreme initial conditions, and strict compliance with prescribed thrust and torque bounds.



\appendices

\section{Hybrid Dynamic Path-lifting Algorithm}
\label{appendix:PathLiftingAlgorithm}

Let $\mathbf{x}_{\mathbf{l}} = (\mathbf{\hat{q}}, m^*, \mathbf{\Tilde{R}}) \in \boldsymbol{\chi}_{\mathbf{l}} \coloneqq \mathbb{S}^3 \times \{-1,1\} \times \mathrm{SO}(3)$, with $m^* \in \{-1,1\}$ and $\delta \in \mathbb{R}_{>0}$ denoting, respectively, a discrete state and a hysteresis margin, and define
\vspace{-0.1500cm}
\begin{equation*}
    \mathbf{C_m} \!\coloneqq\! \{ \mathbf{x_l} \in \boldsymbol{\chi}_{\mathbf{l}}  :  \|\boldsymbol{\varphi}(m^*\boldsymbol{\Phi}(\mathbf{\hat{q}}, \mathbf{\Tilde{R}}))\| \leq 1 + \delta\},
\vspace{-0.1500cm}
\end{equation*}
\begin{equation*}
    \mathbf{D_m} \!\coloneqq\! \{ \mathbf{x_l} \in \boldsymbol{\chi}_{\mathbf{l}}  :  \|\boldsymbol{\varphi}(m^*\boldsymbol{\Phi}(\mathbf{\hat{q}}, \mathbf{\Tilde{R}}))\| \geq 1 + \delta\},
\vspace{-0.1500cm}
\end{equation*}
where $\boldsymbol{\Phi}:\mathbb{S}^3 \times \mathrm{SO}(3) \rightrightarrows \mathbb{S}^3$ represents the map \cite{mayhew2013}
\vspace{-0.60cm}
\begin{equation*}
    \boldsymbol{\Phi}(\mathbf{\hat{q}}, \mathbf{\Tilde{R}}) \coloneqq \!\!\begin{array}{cc}
        & \\ \textrm{argmax} &  \!\mathbf{\hat{q}}^{\!\top}\mathbf{p} \\ 
         \mathbf{p} \in \mathcal{Q}(\mathbf{\Tilde{R}}) & 
    \end{array},
\vspace{-0.1500cm}
\end{equation*}
and $\mathbf{\hat{q}} \coloneqq (\hat{q}_0, \mathbf{\hat{q}_1}) \in \mathbb{S}^3 $ is a memory state. Define also 
\vspace{-0.1500cm}
\begin{equation*}
    \mathbf{C_q} \coloneqq \{\mathbf{x_l} \in \boldsymbol{\chi}_{\mathbf{l}}:  \mathrm{dist}(\mathbf{\hat{q}}, \mathcal{Q}(\mathbf{\Tilde{R}}) \!) \!\leq \!\alpha\}, 
\vspace{-0.1500cm}
\end{equation*}
\begin{equation*}
    \mathbf{D_q} \coloneqq \{\mathbf{x_l} \in \boldsymbol{\chi}_{\mathbf{l}}:  \mathrm{dist}(\mathbf{\hat{q}}, \mathcal{Q}(\mathbf{\Tilde{R}}) \!) \!\geq \!\alpha\}, 
\vspace{-0.10cm}
\end{equation*}
with $\mathrm{dist}(\mathbf{\hat{q}}, \mathcal{Q}(\mathbf{\Tilde{R}})\!) = \mathrm{inf} \{\!1 \!-\! \mathbf{\hat{q}}^{\!\top} \! \mathbf{p}: \mathbf{p} \!\in\! \mathcal{Q}(\mathbf{\Tilde{R}})\}$ and $\alpha \!\in\! (0,1)$. Then, the autonomous hybrid system $\mathcal{H}_{\mathbf{l}}\coloneqq \left(\mathbf{C}_{\mathbf{l}},\;\mathbf{F}_{\mathbf{l}},\;\mathbf{D}_{\mathbf{l}},\;\mathbf{G}_{\mathbf{l}}\right)$, with the state $\mathbf{x_l} \in \boldsymbol{\chi}_{\mathbf{l}}$, the data  
\vspace{-0.1500cm}
\begin{subequations}
    \begin{equation*}
        \mathbf{C}_{\mathbf{l}} \!\coloneqq\! \mathbf{C_q} \cap \!\;\mathbf{C_m}, \quad \!\! \mathbf{F}_{\mathbf{l}}(\mathbf{x_l}) \!\coloneqq\!\! \left(\!\!\! \begin{array}{cc}
             \boldsymbol{0}   \\
             0   \\
             \mathbf{\Tilde{R}}\left[K_\omega^*\mathbb{B}^3\right]   \\
        \end{array}\!\!\! \right)\!\!, \quad \!\!\!  \mathbf{D}_{\mathbf{l}} \!\coloneqq\! \mathbf{D_q} \cup \!\; \mathbf{D_m}
\vspace{-0.1500cm}
    \end{equation*}
 \begin{equation*}
        \mathbf{G_l} (\mathbf{x_l}) \!\coloneqq\! \! \left\{ \!\!\!\! \begin{array}{ll}
             (\boldsymbol{\Phi}(\mathbf{\hat{q}}, \mathbf{\Tilde{R}}), m^*\!, \mathbf{\Tilde{R}})\!\!\!\!\!\! &, \; \mathbf{x_l} \!\in\! \mathbf{D_q} \!\setminus\! \mathbf{D_m} \\
          \{(\boldsymbol{\Phi}(\mathbf{\hat{q}}, \mathbf{\Tilde{R}}), m^*\!, \mathbf{\Tilde{R}}), (\mathbf{\hat{q}}, -m^*\!, \mathbf{\Tilde{R}})\} \!\!\!\!\!\! &, \; \mathbf{x_l} \!\in\! \mathbf{D_q} \!\cap\! \mathbf{D_m} \\   
         (\mathbf{\hat{q}}, -m^*\!, \mathbf{\Tilde{R}}) \!\!\!\!\!\! &, \; \mathbf{x_l} \!\in\! \mathbf{D_m} \!\setminus\! \mathbf{D_q}
        \end{array}\right.
        \vspace{-0.1500cm}
    \end{equation*}
\end{subequations}
\noindent where $\mathbf{\dot{\Tilde{R}}} \in \mathbf{\Tilde{R}}\left[K_\omega^*\mathbb{B}^3\right]$ governs the dynamics of trajectories $\mathbf{\Tilde{R}}:\mathbb{R}_{\geq} 0 \to \mathrm{SO}(3)$ for some $K_\omega^* \in \mathbb{R}_{>0}$, and the output 
\vspace{-0.10cm}
\begin{equation*}
    \boldsymbol{\Tilde{\vartheta}} \coloneqq \left\{ \!\!\!\begin{array}{cl}
        \boldsymbol{\varphi}(m\boldsymbol{\Phi}(\mathbf{\hat{q}}, \mathbf{\Tilde{R}})) &, \quad \mathbf{x}_{\mathbf{l}} \in \mathbf{C_l} \\
         \emptyset &, \quad \mathbf{x}_{\mathbf{l}} \notin \mathbf{C_l}
    \end{array}\right. ,
\vspace{-0.1500cm}
\end{equation*}
represents the path-lifting mechanism from $\mathrm{SO}(3)$ to $(1+\delta)\mathbb{B}^3$. For further details, please see \cite{martins2025hybrid,martins2024}.

\bibliographystyle{IEEEtran}
\bibliography{IEEEabrv,LCSS.bib}

\begin{thebibliography}{10}
\providecommand{\url}[1]{#1}
\csname url@samestyle\endcsname
\providecommand{\newblock}{\relax}
\providecommand{\bibinfo}[2]{#2}
\providecommand{\BIBentrySTDinterwordspacing}{\spaceskip=0pt\relax}
\providecommand{\BIBentryALTinterwordstretchfactor}{4}
\providecommand{\BIBentryALTinterwordspacing}{\spaceskip=\fontdimen2\font plus
\BIBentryALTinterwordstretchfactor\fontdimen3\font minus \fontdimen4\font\relax}
\providecommand{\BIBforeignlanguage}[2]{{%
\expandafter\ifx\csname l@#1\endcsname\relax
\typeout{** WARNING: IEEEtran.bst: No hyphenation pattern has been}%
\typeout{** loaded for the language `#1'. Using the pattern for}%
\typeout{** the default language instead.}%
\else
\language=\csname l@#1\endcsname
\fi
#2}}
\providecommand{\BIBdecl}{\relax}
\BIBdecl

\bibitem{aguiar2007trajectory}
A.~P. Aguiar and J.~P. Hespanha, ``Trajectory-tracking and path-following of underactuated autonomous vehicles with parametric modeling uncertainty,'' \emph{IEEE Transactions on Automatic Control}, vol.~52, no.~8, pp. 1362--1379, 2007.

\bibitem{BhatBernstein2000}
S.~Bhat and D.~Bernstein, ``A topological obstruction to continuous global stabilization of rotational motion and the unwinding phenomenon,'' \emph{Systems \& Control Letters}, vol.~39, no.~1, pp. 63--70, 2000.

\bibitem{mayhew2011topological}
C.~G. Mayhew and A.~R. Teel, ``On the topological structure of attraction basins for differential inclusions,'' \emph{Systems \& Control Letters}, vol.~60, no.~12, pp. 1045--1050, 2011.

\bibitem{casau2015}
P.~Casau, R.~G. Sanfelice, R.~Cunha, D.~Cabecinhas, and C.~Silvestre, ``Robust global trajectory tracking for a class of underactuated vehicles,'' \emph{Automatica}, vol.~58, pp. 90--98, 2015.

\bibitem{Naldi2017}
R.~{Naldi}, M.~{Furci}, R.~G. {Sanfelice}, and L.~{Marconi}, ``Robust global trajectory tracking for underactuated vtol aerial vehicles using inner-outer loop control paradigms,'' \emph{IEEE Transactions on Automatic Control}, vol.~62, no.~1, pp. 97--112, 2017.

\bibitem{wang2021hybrid}
M.~Wang and A.~Tayebi, ``Hybrid feedback for global tracking on matrix lie groups $\mathrm{SO}(3) $ and $\mathrm{SE}(3) $,'' \emph{IEEE Transactions on Automatic Control}, vol.~67, no.~6, pp. 2930--2945, 2021.

\bibitem{basso2022globaluav}
E.~A. Basso, H.~M. Schmidt-Didlaukies, and K.~Y. Pettersen, ``Global asymptotic position and heading tracking for multirotors using tuning function-based adaptive hybrid feedback,'' \emph{IEEE Control Systems Letters}, vol.~7, pp. 295--300, 2022.

\bibitem{martins2023robust}
L.~Martins, C.~Cardeira, and P.~Oliveira, ``Robust global exponential trajectory tracking for quadrotors: An mrp-based hybrid approach,'' in \emph{62nd IEEE Conference on Decision and Control}.\hskip 1em plus 0.5em minus 0.4em\relax IEEE, 2023, pp. 5249--5254.

\bibitem{martins2024}
------, ``Global trajectory tracking for quadrotors: An mrp-based hybrid strategy with input saturation,'' \emph{Automatica}, vol. 162, p. 111521, 2024.

\bibitem{invernizzi2018trajectory}
D.~Invernizzi and M.~Lovera, ``Trajectory tracking control of thrust-vectoring uavs,'' \emph{Automatica}, vol.~95, pp. 180--186, 2018.

\bibitem{dugundji1966topology}
J.~Dugundji, \emph{Topology}, 1st~ed.\hskip 1em plus 0.5em minus 0.4em\relax Allyn and Bacon, 1966.

\bibitem{mayhew2013}
C.~Mayhew, R.~Sanfelice, and A.~Teel, ``On path-lifting mechanisms and unwinding in quaternion-based attitude control,'' \emph{IEEE Transactions on Automatic Control}, vol.~58, no.~5, pp. 1179--1191, 2013.

\bibitem{junkins_2009}
J.~L. Junkins and H.~Schaub, \emph{Analytical Mechanics of Space Systems}.\hskip 1em plus 0.5em minus 0.4em\relax American Institute of Aeronautics and Astronautics, 2009.

\bibitem{khalil_2002}
H.~K. Khalil, \emph{Nonlinear Systems}, 3rd~ed.\hskip 1em plus 0.5em minus 0.4em\relax Prentice Hall, 2002.

\bibitem{martins2025hybrid}
L.~Martins, C.~Cardeira, and P.~Oliveira, ``Hybrid path-lifting algorithm and equivalence of stability results for mrp-based control strategies,'' \emph{IEEE Transactions on Automatic Control}, 2025.

\bibitem{martins2024CDCattitude}
------, ``Robust global attitude tracking on so(3) via mrp-based hybrid feedback,'' in \emph{63rd IEEE Conference on Decision and Control}.\hskip 1em plus 0.5em minus 0.4em\relax IEEE, 2024, pp. 7816--7821.

\bibitem{goebel_2012}
R.~Goebel, R.~G. Sanfelice, and A.~Teel, \emph{Hybrid dynamical systems: modeling, stability, and robustness}.\hskip 1em plus 0.5em minus 0.4em\relax Princeton University Press, 2012.

\bibitem{teel2012}
A.~R. Teel, F.~Forni, and L.~Zaccarian, ``Lyapunov-based sufficient conditions for exponential stability in hybrid systems,'' \emph{IEEE Transactions on Automatic Control}, vol.~58, no.~6, pp. 1591--1596, 2012.

\bibitem{martins2024integrator}
L.~Martins, C.~Cardeira, and P.~Oliveira, ``Global exponential stabilization and global l p performance of a saturated double integrator,'' \emph{IEEE Control Systems Letters}, vol.~8, pp. 784--789, 2024.

\end{thebibliography}

\end{document}